%% file: paper.tex
\documentclass[smallextended]{svjour3}

\usepackage[tight,footnotesize]{subfigure}
\usepackage{booktabs}
\usepackage{url}
\usepackage{natbib}
\usepackage[cmex10]{amsmath}
\usepackage{amssymb}
\usepackage{tikz}
\usepackage{cite}
\usepackage{xspace}
\usetikzlibrary{arrows,shapes,positioning,fit,calc, patterns}
\usepackage{pgfplots}
\usepackage{pgffor}
\usepackage{ifpdf}

\usepackage[ruled, vlined, nofillcomment, linesnumbered]{algorithm2e}

\usepackage[pdftitle={Fast Sequence Segmentation using Log-Linear Models},
            pdfauthor={Nikolaj Tatti},
			citecolor = black, linkcolor = black, colorlinks = true, bookmarksopen = true, bookmarksnumbered = true, hyperfootnotes = false]{hyperref}

\journalname{Data Mining and Knowledge Discovery}

\input{defines.tex}

\begin{document}
\title{Fast Sequence Segmentation using Log-Linear Models}
\author{Nikolaj Tatti}
\institute{
Nikolaj Tatti \at
Department of Mathematics and Computer Science, University of Antwerp, Antwerp,\\
Department of Computer Science, Katholieke Universiteit Leuven, Leuven,\\
Belgium\\
\email{nikolaj.tatti@gmail.com}
}

\maketitle

\begin{abstract}
Sequence segmentation is a well-studied problem, where given a sequence of
elements, an integer $K$, and some measure of homogeneity, the task is to split
the sequence into $K$ contiguous segments that are maximally homogeneous.  A
classic approach to find the optimal solution is by using a dynamic program.
Unfortunately, the execution time of this program is quadratic with respect to the length
of the input sequence. This makes the algorithm slow for a sequence of
non-trivial length.  In this paper we study segmentations whose measure of
goodness is based on log-linear models, a rich family that contains many of the
standard distributions.  We present a theoretical result allowing us to prune
many suboptimal segmentations.  Using this result, we modify the standard
dynamic program for one-dimensional log-linear models, and by doing so reduce
the computational time.  We demonstrate empirically, that this approach can
significantly reduce the computational burden of finding the optimal
segmentation.
\end{abstract}


\input{introduction.tex}
\input{segmentation.tex}
\input{calculate.tex}
\input{onedimensional.tex}
\input{borders.tex}
\input{experiments.tex}

\input{related.tex}
\input{conclusions.tex}

\section*{Acknowledgements}
Nikolaj Tatti was partly supported by a Post-Doctoral Fellowship of the Research Foundation -- Flanders (\textsc{fwo}).

\bibliographystyle{spbasic}
\bibliography{bibliography}

\appendix
\input{appendix.tex}

\end{document}

%% file: defines.tex
\newcommand{\set}[1]{\left\{#1\right\}}
\newcommand{\pr}[1]{\left(#1\right)}
\newcommand{\fpr}[1]{\mathopen{}\left(#1\right)}

\newcommand{\abs}[1]{{\left|#1\right|}}
\newcommand{\norm}[1]{\left\|#1\right\|}

\newcommand{\enpr}[2]{\pr{#1 ,\ldots , #2}}

\newcommand{\real}{\mathbb{R}}
\newcommand{\funcdef}[3]{{#1}:{#2} \to {#3}}

\newcommand{\freq}[1]{\mathit{av}\fpr{#1}}
\newcommand{\cs}[1]{\mathit{c}\fpr{#1}}

\newcommand{\intl}[1]{\mathit{int}_L\fpr{#1}}
\newcommand{\intr}[1]{\mathit{int}_R\fpr{#1}}

\newcommand{\score}[1]{\mathit{sc}\fpr{#1}}
\newcommand{\diff}[1]{\mathit{diff}\fpr{#1}}
\newcommand{\borders}[1]{\mathit{borders}\fpr{#1}}
\newcommand{\btree}[1]{\mathit{btree}\fpr{#1}}
\newcommand{\children}[1]{\mathit{children}\fpr{#1}}

\newcommand{\define}{\leftarrow}

\newcommand{\updatetree}{\textsc{UpdateTree}\xspace}
\newcommand{\segment}{\textsc{Segment}\xspace}

\SetKw{OR}{or}
\SetKw{AND}{and}
\SetKwInOut{Input}{input\hspace*{0.4cm}}
\SetKwInOut{Output}{output}
\SetArgSty{textnormal}

\renewenvironment{proof}{\begin{oldproof}}{{\hfill \ensuremath{\Box}}\end{oldproof}}

\def\clap#1{\hbox to 0pt{\hss#1\hss}}

\newcommand{\slice}[3]{

  \draw[fill=#3, draw=#3, ultra thin] (0,0) -- (#1:1) arc (#1:#2:1) -- cycle;
}

\pgfdeclarelayer{background}
\pgfdeclarelayer{foreground}
\pgfsetlayers{background,main,foreground}

\definecolor{yafaxiscolor}{rgb}{0.3, 0.3, 0.3}

\definecolor{yafcolor1}{rgb}{0.4, 0.165, 0.553}
\definecolor{yafcolor2}{rgb}{0.949, 0.482, 0.216}
\definecolor{yafcolor3}{rgb}{0.47, 0.549, 0.306}
\definecolor{yafcolor4}{rgb}{0.925, 0.165, 0.224}
\definecolor{yafcolor5}{rgb}{0.141, 0.345, 0.643}
\definecolor{yafcolor6}{rgb}{0.965, 0.933, 0.267}
\definecolor{yafcolor7}{rgb}{0.627, 0.118, 0.165}
\definecolor{yafcolor8}{rgb}{0.878, 0.475, 0.686}

\newlength{\yafaxispad}
\newlength{\yaftlpad}
\newlength{\yaflabelpad}
\newlength{\yafaxiswidth}
\newlength{\yafticklen}
\makeatletter
\def\pgfplots@drawtickgridlines@INSTALLCLIP@onorientedsurf#1{}
\makeatother

\newcommand{\yafdrawaxis}[4]{
	\pgfplotstransformcoordinatex{#1}\let\xmincoord=\pgfmathresult 
	\pgfplotstransformcoordinatex{#2}\let\xmaxcoord=\pgfmathresult 
	\pgfplotstransformcoordinatey{#3}\let\ymincoord=\pgfmathresult 
	\pgfplotstransformcoordinatey{#4}\let\ymaxcoord=\pgfmathresult 
	\pgfsetlinewidth{\yafaxiswidth} 
	\pgfsetcolor{yafaxiscolor}
	\pgfpathmoveto{\pgfpointadd{\pgfpointadd{\pgfplotspointrelaxisxy{0}{0}}{\pgfqpointxy{\xmincoord}{0}}}{\pgfqpoint{-0.5\yafaxiswidth}{\yafaxispad}}}
	\pgfpathlineto{\pgfpointadd{\pgfpointadd{\pgfplotspointrelaxisxy{0}{0}}{\pgfqpointxy{\xmaxcoord}{0}}}{\pgfqpoint{0.5\yafaxiswidth}{\yafaxispad}}}
	\pgfpathmoveto{\pgfpointadd{\pgfpointadd{\pgfplotspointrelaxisxy{0}{0}}{\pgfqpointxy{0}{\ymincoord}}}{\pgfqpoint{\yafaxispad}{-0.5\yafaxiswidth}}}
	\pgfpathlineto{\pgfpointadd{\pgfpointadd{\pgfplotspointrelaxisxy{0}{0}}{\pgfqpointxy{0}{\ymaxcoord}}}{\pgfqpoint{\yafaxispad}{0.5\yafaxiswidth}}}
	\pgfusepath{stroke}
}

\pgfplotscreateplotcyclelist{yaf}{%
{yafcolor1,mark options={scale=0.75},mark=o}, 
{yafcolor2,mark options={scale=0.75},mark=square},
{yafcolor3,mark options={scale=0.75},mark=triangle},
{yafcolor4,mark options={scale=0.75},mark=o},
{yafcolor5,mark options={scale=0.75},mark=o},
{yafcolor6,mark options={scale=0.75},mark=o},
{yafcolor7,mark options={scale=0.75},mark=o},
{yafcolor8,mark options={scale=0.75},mark=o}} 

\pgfplotsset{axis y line=left, axis x line=bottom,
	tick align=outside,
	tickwidth=\yafticklen,
	clip = false,
    x axis line style= {-, line width = 0pt, color=black!0},
    y axis line style= {-, line width = 0pt, color=black!0},
    x tick style= {line width = \yafaxiswidth, color=yafaxiscolor, yshift = \yafaxispad},
    y tick style= {line width = \yafaxiswidth, color=yafaxiscolor, xshift = \yafaxispad},
    x tick label style = {font=\scriptsize, yshift = \yaftlpad},
    y tick label style = {font=\scriptsize, xshift = \yaftlpad},
    every axis y label/.style = {at = {(ticklabel cs:0.5)}, rotate=90, anchor=center, font=\scriptsize, yshift = -\yaflabelpad},
    every axis x label/.style = {at = {(ticklabel cs:0.5)}, anchor=center, font=\scriptsize, yshift = \yaflabelpad},
    x tick label style = {font=\scriptsize, yshift = 1pt},
    grid = major,
    major grid style  = {dash pattern = on 1pt off 3 pt},
	every axis plot post/.append style= {line width=\yafaxiswidth} ,
	legend cell align = left,
	legend style = {inner sep = 1pt, cells = {font=\scriptsize}},
	legend image code/.code={%
		\draw[mark repeat=2,mark phase=2,#1] 
		plot coordinates { (0cm,0cm) (0.15cm,0cm) (0.3cm,0cm) };%
	} 
}

%% file: introduction.tex
\section{Introduction}
\label{sec:introduction}
Sequence segmentation is a well-studied problem, where given a sequence of
elements, an integer $K$, and some measure of homogeneity, the task is to split
the sequence into $K$ contiguous segments that are maximally homogeneous.

An exact solution for segmentation with $K$ segments can be obtained by a classic dynamic
program in $O(L^2K)$ time, where $L$ is the length of the
sequence~\citep{bellman:61:on}. Due to the quadratic complexity, we cannot apply
segmentation for sequences of non-trivial length. 
In this paper we introduce a speedup to the dynamic program used for solving the
exact solution. Our key result, given in Theorem~\ref{thr:suff}, states that
when certain conditions are met, we can discard the candidate for a segment border,
thus speeding up the inner loop of the dynamic program.

We consider segmentation using the log-likelihood of a
log-linear model to score the goodness of individual segments. Many standard
distributions can be described as log-linear models, including  Bernoulli, Gamma,
Poisson, and Gaussian distributions. Moreover, when using a Gaussian
distribution, optimizing the log-likelihood is equal to the minimizing the
$L_2$ error (see Example~\ref{ex:gaussian}).

The conditions given in Theorem~\ref{thr:suff} are hard to verify, however, we
demonstrate that this can be done with relative ease for one-dimensional
models. The key idea is as follows: Consider segmenting the sequence given in
Figure~\ref{fig:toy:a} into $2$ segments using the $L_2$ error. Assume a segmentation $[1,
100]$, $[101, 200]$. Figure~\ref{fig:toy:b} tells us that this segmentation is not
optimal.
In fact, the optimal segmentation with 2 segments for this data is $[1, 70], [71, 200]$.

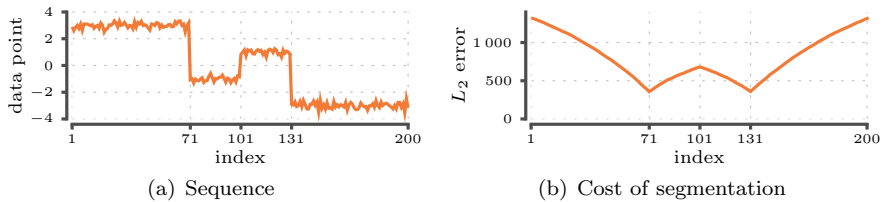
\begin{figure}[htb!]
\begin{center}
\subfigure[Sequence\label{fig:toy:a}]{
\begin{tikzpicture}
\begin{axis}[xlabel=index, ylabel= {data point},
	width = 6cm,
	height = 3cm,
	xmin = 1,
	xmax = 200,
	ymin = -4,
	ymax = 4,
	xtick = {1, 71, 101, 131, 200},
	scaled x ticks = false,
	cycle list name=yaf,
	xticklabel = {$\scriptstyle\pgfmathprintnumber{\tick}$},
	yticklabel = {$\scriptstyle\pgfmathprintnumber{\tick}$}
	]

\addplot[yafcolor2, no markers] table[x expr = {\lineno + 1}, y index = 0, header = false] {toy.dat};

\pgfplotsextra{\yafdrawaxis{1}{200}{-4}{4}}
\end{axis}
\end{tikzpicture}}
\subfigure[Cost of segmentation\label{fig:toy:b}]{
\begin{tikzpicture}
\begin{axis}[xlabel=index, ylabel= {$L_2$ error},
	width = 6cm,
	height = 3cm,
	xmin = 1,
	xmax = 200,
	ymin = 0,
	ymax = 1400,
	xtick = {1, 71, 101, 131, 200},
	scaled x ticks = false,
	cycle list name=yaf,
	y tick label style = {/pgf/number format/set thousands separator = {\,}},
	xticklabel = {$\scriptstyle\pgfmathprintnumber{\tick}$},
	yticklabel = {$\scriptstyle\pgfmathprintnumber{\tick}$}
	]

\addplot[yafcolor2, no markers] table[x expr = {\lineno + 1}, y index = 0, header = false] {toycost.dat};

\pgfplotsextra{\yafdrawaxis{1}{200}{0}{1400}}
\end{axis}
\end{tikzpicture}}

\end{center}
\caption{Toy sequence and the $L_2$ cost of a segmentation $[1, k - 1]$, $[k, 200]$ as a
function of $k$. 
In this paper we propose a necessary condition for a segmentation to be optimal.
This condition allows us to prune suboptimal segmentations, such as $[1, 100], [101, 200]$}
\label{fig:toy}
\end{figure}

Sequence values around $101$ have a particular characteristic which we can exploit
to speedup the optimization. In order to demonstrate this, let us define 
\[
\begin{split}
	X & = \big\{\frac{1}{101 - j} \sum_{i = j}^{100} D_i \mid 1 \leq j \leq 100\big\} \quad \text{and} \\
	Y & = \big\{\frac{1}{j - 100} \sum_{i = 101}^{j} D_i \mid 101 \leq j \leq 200\big\},
\end{split}
\]
that is, $X$ contains the averages from the right side of the first segment and $Y$
contains the averages from the left side of the second segment.
Let us define $r_1  = \min X$, $r_2 = \max X$, $l_1  = \min Y$, $l_2 = \max Y$.
We see that $r_1 \approx -1$, $r_2 \approx
1.8$, $l_1 \approx -1.8$, and $l_2 \approx 1$. That is, the intervals $[r_1,
r_2]$ and $[l_1, l_2]$ intersect.  We will show in such case that not only we can safely ignore the segmentation $[1, 100]$, $[101, 200]$
but we also will show that even if we augment the sequence with additional data points,
index $101$ will never be part of the optimal segmentation with 2 segments.
This pruning allows us to speedup the dynamic programming.

In general, if the extreme values of averages $[r_1, r_2]$ and $[l_1, l_2]$
computed from neighboring segments intersect, we know that
the segmentation is suboptimal.  On the other hand, the optimal segmentation with
$4$ segments for data in Figure~\ref{fig:toy:a}, $[1, 70], [71, 100], [101,
130], [131, 200]$, uses index $101$.  We do not violate our condition since the extreme values of
averages $[r_1, r_2]$ computed \emph{only} from the second segment and extreme
values of averages $[l_1, l_2]$ computed \emph{only} from the third segment no
longer intersect.

Using this idea, we will build an efficient pruning technique for segmenting 
data using one-dimensional log-linear models.  We empirically demonstrate that this
approach can reduce the computational load by several orders of magnitude
compared to the standard approach.

The remaining paper is organized as follows. In Section~\ref{sec:segmentation}
we give preliminary notation and define the segmentation problem. In
Section~\ref{sec:calculate} we give the key result which allows us to prune
segments. Sections~\ref{sec:onedim}--\ref{sec:border} are devoted to a
segmentation algorithm. We present our experiments in
Section~\ref{sec:experiments} and related work in Section~\ref{sec:related}.
Finally, we conclude the paper with discussion in
Section~\ref{sec:conclusions}.

%% file: segmentation.tex
\section{Segmentation for log-linear models}
\label{sec:segmentation}
In this section we give preliminaries and define the segmentation problem.

A \emph{sequence} $D = \enpr{D_1}{D_L}$ is a sequence of real vectors of length
$M$, $D_i \in \real^M$. A \emph{segment} $I = [b, e]$ consists of two integers
such that $1 \leq b \leq e \leq L$.
We will write $k \in I$ whenever $b \leq k \leq e$ for an integer $k$.
We define $D[b, e] = \enpr{D_b}{D_e}$ to
be the subsequence corresponding to that segment. A \emph{segmentation} $P$ is
a list of disjoint segments that cover $D$, that is, $P = \enpr{I_1}{I_K}$
such that the first segment $I_1$ starts
at $1$, the last segment $I_K$ ends at $\abs{D}$ and $I_k = [a, b]$ begins right after 
$I_{k - 1} = [c, d]$, that is $a = d + 1$.

Our goal is to find a segmentation that maximizes the likelihood of a
log-linear model of each individual segment. By log-linear models, also known as exponential family, we mean
models whose probability density function can we written as
\[
	p(x \mid r) = q(x)\exp\fpr{Z(r) + r^TS(x)},
\]
where $\funcdef{S}{\real^M}{\real^N}$ is a function mapping $x$ to a vector in
$\real^N$, $r \in \real^N$ is the parameter vector of the model, and $Z(r)$ is
the normalization constant. Many standard distributions are log-linear, for example,
Poisson, Gamma, Bernoulli, Binomial, and Gaussian (both with fixed or unknown variance). 
We will argue later in this section that using a Gaussian distribution with a fixed variance is equivalent to minimizing
$L_2$ error.

Assume that we are given a segmentation $P$ and for each segment $I \in P$,
we have a parameter vector $r_I$. Let us now consider the log-likelihood
\[
\begin{split}
	\log \prod_{I \in P} \prod_{k \in I} p(D_k \mid r_I) & = \sum_{I \in P} \sum_{k \in I} \log q(D_k) + Z(r_I) + r_I^TS(D_k)\\
	& = \sum_{k = 1}^{\abs{D}} \log q(D_k) +  \sum_{I \in P} \sum_{k \in I} Z(r_I) + r_I^TS(D_k),\\
\end{split}
\]
for this segmentation of $D$.
Note that the first term in the right-hand side does not depend on the
parameters nor on the segmentation. Consequently, we can ignore it.  In
addition, note that we can safely assume that $S(x) = x$. If this is not the case, we can always transform sequence $D$
into $D' = \enpr{S(D_1)}{S(D_L)}$. From now on we will assume that $S(x) = x$.

For notational simplicity, let us define  
\[
	\cs{D} = \sum_{i = 1}^{\abs{D}} D_i \quad\text{and}\quad \freq{D} = \frac{\cs{D}}{\abs{D}}
\]
to be the sum and the average of data points in $D$. If $D$ is clear from the
context, we will often write $\cs{i, j}$ and $\freq{i, j}$ to mean $\cs{D[i,
j]}$ and $\freq{D[i, j]}$.  As shorthand, we write $\freq{j}$ and $\cs{j}$ to mean
$\cs{1, j}$ and $\freq{1, j}$.

We define the score of a single segment given a parameter vector as
\[
	\score{D \mid r} = \abs{D}Z(r) + r^T\cs{D}\quad.
\]
We define the score for a segmentation $P$ as
\[
	\score{P; D} = \sum_{I \in P} \score{D[I]}, \text{ where } \score{D} = \sup_r \score{D \mid r},
\]
that is, $\score{P; D}$ is a sum of the optimal scores of individual segments.
We see that optimizing $\score{P; D}$ is equivalent to maximizing likelihood
of the log-linear model.

We are now ready to state our optimization problem.
\begin{problem}
Given a sequence $D$, a log-linear model, and an integer $K$, find a segmentation
$P$ with $K$ segments maximizing $\score{P; D}$. 
\end{problem}

\begin{example}
\label{ex:gaussian}
Let us now consider a Gaussian distribution with identity covariance matrix. This
distribution is log-linear since we can rewrite 
\[
(2\pi)^{-M/2}e^{-0.5\norm{x - \mu}^2} \quad\text{as}\quad
	e^{-0.5\norm{x}^2} (2\pi)^{-M/2} e^{-0.5\norm{\mu}^2 + \mu^Tx}\quad.
\]
The log-likelihood of a Gaussian distribution for a segmentation $P$ is
\[
\begin{split}
	&\log \prod_{I \in P} \prod_{x \in D[I]} (2\pi)^{-M/2} e^{-0.5\norm{x - \mu_I}^2}\\
	&\quad = \sum_{I \in P} \sum_{x \in D[I]} -M/2 \log 2\pi -0.5\norm{x - \mu_I}^2\\
	&\quad = -\abs{D}M/2 \log 2\pi -0.5  \sum_{I \in P} \sum_{x \in D[I]} \norm{x - \mu_I}^2\quad.\\
\end{split}
\]
The optimal value for $\mu_I$ is an average of data points in $D[I]$.
The first term of the right-hand side is constant while the second term is the $L_2$ error.
Consequently, selecting a segmentation that maximizes log-likelihood is equivalent to
finding a segmentation that minimizes the $L_2$ error, a typical choice for an error function.
\end{example}

The optimal segmentation can be found with a dynamic program~\citep{bellman:61:on}. In order
to see this, let $P = \enpr{I_1}{I_K}$ be the optimal segmentation with $K$
segments. Let $c$ be the last index of $I_{K - 1}$.  Then $\enpr{I_1}{I_{K -
1}}$ is the optimal segmentation for $D[1, c]$. We can find the optimal
segmentation with $K$ segments by first computing the optimal segmentation with
$K - 1$ segments for each $D[1, c]$ and then testing which segment of form $(c, \abs{D})$
we need to add to produce the optimal segmentation with $K$ segments. This leads
to an algorithm of time complexity $O(K\abs{D}^2)$.  The goal of this paper
is to provide an optimization of this dynamic program.

%% file: calculate.tex
\section{Necessary Condition for Optimal Segmentation}
\label{sec:calculate}
In this section we give a key result of this paper. This result allows us
to prune candidates that will not be included in the optimal segmentation and hence speedup the 
dynamic program.

In order to do so, let $V$ be a set of vectors in $\real^N$.  We say that $V$
is a \emph{cover} if for any $y \in \real^N$, there is a $v \in V$ such that
$y^Tv \geq 0$. See Figure~\ref{fig:cover} for an example. Given two sequences
$D$ and $E$ we define $\diff{D, E}$ to be the \emph{difference set} for $D$ and $E$ as 
\[
	\diff{D, E}  = \set{\freq{D[k, \abs{D}]} - \freq{E[1, l]} \mid 1 \leq k \leq \abs{D}, 1 \leq l \leq \abs{E}}\quad.
\]

\begin{figure}[htb!]
\begin{center}
\hspace{\fill}
\subfigure[non-cover\label{fig:cover:a}]{
\begin{tikzpicture}[thick, > = latex]
\draw[step=0.5cm, black!20, ultra thin, dashed] (-1.4, -1.4) grid (1.4, 1.4);
\draw[gray] (-1.5, 0) -- (1.5, 0);
\draw[gray] (0, -1.5) -- (0, 1.5);
\begin{pgfonlayer}{background}
\slice{-60}{120}{yafcolor1!20}
\slice{40}{220}{yafcolor2!20}
\slice{40}{120}{yafcolor1!50!yafcolor2!30}
\end{pgfonlayer}{background}
\node[circle] at (canvas polar cs:angle=30, radius=1.4cm) (x1) {$x_1$};
\node[circle] at (canvas polar cs:angle=130, radius=1.4cm) (x2) {$x_2$};
\node[circle] at (canvas polar cs:angle=260, radius=1.4cm) (x3) {$u$};
\draw[->, yafcolor1] (0, 0) -- (x1);
\draw[->, yafcolor2] (0, 0) -- (x2);
\draw[->] (0, 0) -- (x3);
\end{tikzpicture}}
\hfill
\subfigure[cover\label{fig:cover:b}]{
\begin{tikzpicture}[thick, > = latex]
\draw[step=0.5cm, black!20, ultra thin, dashed] (-1.4, -1.4) grid (1.4, 1.4);
\draw[gray] (-1.5, 0) -- (1.5, 0);
\draw[gray] (0, -1.5) -- (0, 1.5);
\begin{pgfonlayer}{background}
\slice{-60}{120}{yafcolor1!20}
\slice{40}{220}{yafcolor2!20}
\slice{150}{330}{yafcolor3!20}
\slice{40}{120}{yafcolor1!50!yafcolor2!30}
\slice{150}{220}{yafcolor2!50!yafcolor3!30}
\slice{300}{330}{yafcolor1!50!yafcolor3!30}
\end{pgfonlayer}{background}
\node[circle] at (canvas polar cs:angle=30, radius=1.4cm) (x1) {$x_1$};
\node[circle] at (canvas polar cs:angle=130, radius=1.4cm) (x2) {$x_2$};
\node[circle] at (canvas polar cs:angle=240, radius=1.4cm) (x3) {$x_3$};
\draw[->, yafcolor1] (0, 0) -- (x1);
\draw[->, yafcolor2] (0, 0) -- (x2);
\draw[->, yafcolor3] (0, 0) -- (x3);
\end{tikzpicture}}
\hspace{\fill}

\end{center}
\caption{An example of a non-cover and a cover. In Figure~\ref{fig:cover:a} $\set{x_1, x_2}$ is not a cover since $u$
is outside the half-planes induced by $x_1$ and $x_2$. In Figure~\ref{fig:cover:b} $\set{x_1, x_2, x_3}$ is a cover}
\label{fig:cover}
\end{figure}
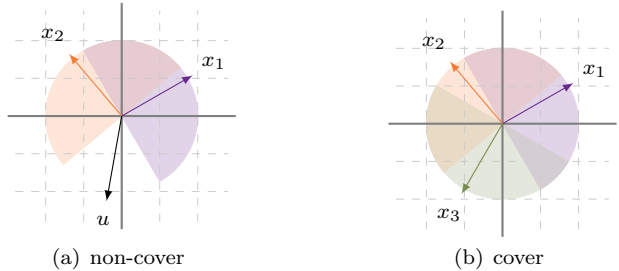

We are now ready to state the key result of the paper. For readability, we postpone the proof to Appendix~\ref{sec:appprune}.
\begin{theorem}
\label{thr:suff}
Let $P$ be a segmentation. There is a segmentation $P'$ such that $\score{P'}
\geq \score{P}$ and $\diff{D[I], D[J]}$ is not a cover for any two consecutive
segments $I$ and $J$ in $P'$.
\end{theorem}

%% file: onedimensional.tex
\section{Segmentation for one-dimensional models}
\label{sec:onedim}
In the previous section we saw a necessary condition for optimal segmentation.
This involves checking whether the difference set of consecutive segments is
a cover. In this section and the next section we show that we can efficiently check this
condition if our linear model is one-dimensional, that is,
if data points $D_i$ are real numbers.

In order to show this, let $D$ be a sequence.  We define a \emph{left
interval} to be an interval
\[
	\intl{D} = \big(\min_{1 \leq i \leq \abs{D}} \freq{{1, i}}, \max_{1 \leq i \leq \abs{D}} \freq{{1, i}}\big)
\]
of extreme values of $\freq{1, i}$.
Similarly,
we define a \emph{right interval} to be 
\[
	\intr{D} = \big(\min_{1 \leq i \leq \abs{D}} \freq{{i, \abs{D}}}, \max_{1 \leq i \leq \abs{D}} \freq{{i, \abs{D}}}\big)\quad.
\]

We can now express the condition using these intervals.
\begin{theorem}
\label{thr:minmax}
Assume two sequences, $D$ and $E$ and let $S$ be a one-dimensional statistic.
Then $\diff{D, E}$ is a cover if and only if the intervals $\intr{D}$ and $\intl{E}$
intersect.
\end{theorem}

\begin{proof}
Let $\intr{D} = (x, y)$ and $\intl{E} = (u, v)$.  $\diff{D, E}$ is a cover if
and only if there are $a$, $b$, $c$, and $d$ such that $\freq{D[a, \abs{D}]}
\leq \freq{E[1, b]}$  and  $\freq{D[c, \abs{D}]} \geq \freq{E[1, d]}$.  This
is equivalent to $x \leq v$ and $y \geq u$, which is equivalent to $\intr{D}$
and $\intl{E}$ intersecting.
\end{proof}

We can now use this result to design an efficient algorithm.  Assume that we
already have computed for each $j$ the optimal segmentation with $K - 1$ segments, say $P_j$
covering $D[1, j]$. We now want to find an optimal segmentation with $K$ segments 
covering $D[1, i]$. In order to do so we need to augment each $P_{j - 1}$ for $j \leq i$
with a segment $[j, i]$, and pick the optimal segmentation. Assume that
the intervals $\intl{D[j, i]}$ and $\intr{D[c, j - 1]}$ intersect,
where $c$ is the starting point of the last segment in $P_{j - 1}$.
Then Theorems~\ref{thr:suff}~and~\ref{thr:minmax} imply
that we can safely ignore the segmentation $P_{j - 1}$ augmented with $[j, i]$.
Moreover, if the intervals intersect when segmenting $D[1, i]$, they will also
intersect when segmenting $D[1, k]$ for $k > i$. Hence, as soon as $\intl{D[j, i]}$ and $\intr{D[c, j - 1]}$
intersect, we can ignore $j$ as a candidate for the starting point of the
last segment. We present the pseudo-code for this approach as Algorithm~\ref{alg:fast}.

\begin{algorithm}[htb!]
\caption{$\textsc{Segment}(s, r, D)$ builds the optimal $K$-segmentation for $D$ using optimal segmentations with $K - 1$ segments}
\Input{scores $s$ for optimal segmentation with $K - 1$ segments, corresponding right intervals $r$, sequence $D$}
\Output{scores $u$ for optimal $K$-segmentation, right intervals $v$ for optimal $K$-segmentation}
\label{alg:fast}

$C \define \emptyset$\;
\ForEach {$i = 1, \ldots, \abs{D}$} {
	add $i$ to $C$\;\nllabel{alg:fast:add}
	\ForEach {$j \in C$} {
		update $l(j)$ to be $\intl{D[j, i]}$\;
		\If {$r(j - 1)$ and $l(j)$ intersect} {
			delete $j$ from $C$\;\nllabel{alg:fast:delete}
		}
	}

	$c \define \arg \max_{j \in C}  s(j - 1) + \score{D[j, i]}$\;\nllabel{alg:fast:inner}
	$u(i) \define s(c - 1) + \score{D[c, i]}$\;
	$v(i) \define \intr{D[c, i]}$\;
}
\Return $u$, $v$\;
\end{algorithm}

Let us next analyze the time and memory complexity of Algorithm~\ref{alg:fast}.
Let $L$ be the maximal size of $C$.
It is easy to see that we can compute $\score{D[j,
i]}$ and $\intl{D[j, i]}$ in constant time by keeping and updating the sum
$\cs{j, i}$ for every $j \in C$. The only non-trivial part of
Algorithm~\ref{alg:fast} is computing the right interval $\intr{D[c, i]}$.
In the next section, we will show how to compute the right interval in
amortized $O(L)$ time, hence the execution time of the algorithm is in
$O(L\abs{D})$. Moreover, we will show that the total memory requirement for
computing the right interval is in $O(\abs{D})$ which will make the memory
usage of the algorithm $O(\abs{D})$.

%% file: borders.tex
\section{Computing the Right Interval}
\label{sec:border}
In this section we show how to compute the right interval, as needed in
Algorithm~\ref{alg:fast}. We will focus on how to compute the maximal value of the
right interval; we can compute the minimal value using exactly the same framework.

\subsection{Computing the Borders}

Our first goal, given a sequence $D$ and integer $i$, is to find $j$ such that
$\freq{j, i}$ is maximal.  Naturally, if we have to do so from scratch, we have
no other option but to test every $1 \leq j \leq i$. However, since segmentation
needs the maximal average for every $i$ we can use information from previous scans
to find the optimal $j$ more quickly.

We will now present the main results by~\citet{calders:07:mining} in which the authors
considered efficiently finding the maximal average from a stream of data
points. In the next section we will modify this approach to make it more memory-efficient.

Given a sequence $D$ we say that $1 \leq i \leq \abs{D}$ is a \emph{border}
if there is a (possibly empty) sequence $E$ such that if we define $F$ to be $D$ concatenated
with $E$, then
\[
	\freq{F[i, \abs{F}]} = \max_{1 \leq j \leq \abs{F}} \freq{F[j, \abs{F}]}\quad.
\]
We define $\borders{D}$ to be the sorted list of border points of $D$.

Let $D$ be a sequence. Further, Let $1 \leq i \leq j \leq \abs{D}$ and let $\enpr{b_1}{b_M} =
\borders{D[i, j]}$. Whenever $D$ is clear from the context,  we define $\borders{i, j}$ to be $\enpr{b_1 + i -
1}{b_M + i - 1}$. Further, we will write
$\borders{i}$ instead of $\borders{i, \abs{D}}$.

The following theorem states that a maximal average can be found by simply taking
the largest border.\!\footnote{\citet{calders:07:mining} deal only with binary sequences but we can
easily extend these results to the general case.}
\begin{theorem}[see~\citealp{calders:07:mining}]
\label{thr:bordmax}
Assume a sequence $D$. Let $j = \max \borders{D}$.
Then
\[
	\freq{j, \abs{D}} = \max_{1 \leq k \leq \abs{D}} \freq{k, \abs{D}}\quad.
\]
\end{theorem}

We can describe the borders using the following theorem.

\begin{theorem}[see~\citealp{calders:07:mining}]
\label{thr:bordfreq}
An integer $i$ is a border for $D$ if and only if there are no $a$ and $b$, $a < i \leq b$ such
that $\freq{a, i - 1} \geq \freq{i, b}$.
\end{theorem}

\begin{example}
\label{ex:border:1}
Assume that we are given a sequence $D = \pr{2, 0, 1, 2, 1, 1, 9, 2, 5, 0}$, and that $S(x) = x$.
According to Theorem~\ref{thr:bordfreq}, index $3 \notin \borders{D}$, since $\freq{1, 2} = 1 = \freq{3, 3}$.
The borders are $\pr{1, 4, 7} = \borders{D}$.
\end{example}

Our next step is to revise the algorithm given by~\citet{calders:07:mining} for constructing
$\borders{1, i}$ from $\borders{1, i - 1}$. The key idea for update
is given in the following theorem.

\begin{theorem}[see~\citealp{calders:07:mining}]
\label{thr:update}
Let us assume a list of borders
$\enpr{b_1}{b_M} = \borders{1, i - 1}$.
Define $b_{M + 1} = i$. Define $N$, $2 \leq N \leq M + 1$, to be the maximal integer such that
$\freq{b_{N - 1}, i} < \freq{b_N, i}$. If such $N$ does not exist, we set $N = 1$.
Then, $\enpr{b_1}{b_N} = \borders{1, i}$.
\end{theorem}

The update algorithm (given as Algorithm~\ref{alg:border}) starts with the
previous borders $\enpr{b_1}{b_M} = \borders{1, i - 1}$ and adds $i$ as a
border.  Then the algorithm tests whether the average of the second last
border is larger than the average of the last border. If so, then the
condition in Theorem~\ref{thr:update} is violated, and we remove the
last border and repeat the test. The correctness of \textsc{Update} is given
by~\citet{calders:07:mining}. Note that we can compute the needed averages
in constant time, for example, by precalculating a sequence $\enpr{\cs{1}}{\cs{\abs{D}}}$.

\begin{algorithm}
\caption{\textsc{Update}, updates $\borders{1, i - 1}$ to $\borders{1, i}$}
\Input{borders $\enpr{b_1}{b_M} = \borders{1, i - 1}$}
\Output{updated borders $\borders{1, i}$}
\label{alg:border}

$b_{M + 1} \define i$\;
$M \define M + 1$\;

\lWhile {$M > 1$ \AND $\freq{b_{M - 1}, i} \geq \freq{b_M, i}$} {
	$M \define M - 1$
}

\Return $\enpr{b_1}{b_M}$\;

\end{algorithm}

\subsection{Computing Borders Simultaneously}
We can use borders to discover the right interval for a single segment.
However, recall that in Algorithm~\ref{alg:fast} we need to be able to compute the right
interval for any $D[c, i]$, where $c \in C$ is the current set of candidates
for a segment.  A na\"{i}ve approach would be to compute borders separately for
each $D[c, i]$. This leads to $O(L\abs{D})$ memory and time consumption,
where $L$ is the maximum size of $C$ during evaluation of \textsc{Segment}.  Here, we
will modify the border update algorithm such that its total memory consumption is
$O(\abs{D})$.  This will guarantee that the memory consumption of \textsc{Segment}
is $O(\abs{D})$.

\begin{example}
\label{ex:border:2}
Let us continue Example~\ref{ex:border:1}. We have $\borders{1, 10} = (1, 4, 7)$,
$\borders{3, 10} = (3, 4, 7)$, $\borders{8, 10} = (8, 9)$, and $\borders{10, 10} = (10)$.
Note that $\borders{1, 10}$ and $\borders{3, 10}$ have a common tail sequence, namely, $(4, 7)$.
We can generalize this observation.
\end{example}

The following key result states that when two border lists, say $\borders{i}$
and $\borders{j}$ share a common border, the subsequent borders are equivalent.

\begin{theorem}
\label{thr:tree}
Let $D$ be a sequence and let $1 \leq i, j \leq \abs{D}$ be two indices.
Assume that $a \in \borders{i} \cap \borders{j}$.
Let $b \geq a$. Then $b \in \borders{i}$ if and only if $b \in \borders{j}$.
\end{theorem}

\begin{proof}
We will show that
$\set{b \in \borders{i, k} \mid b \geq a} = \set{b \in \borders{j, k} \mid b \geq a}$
using induction over $k$. The result follows by setting $k = \abs{D}$.

Since $\set{b \in \borders{i, a} \mid b \geq a}  =
\set{a} = \set{b \in \borders{j, a} \mid b \geq a}$, the first $k = a$ step follows.

Assume that the result holds for $k - 1$. Let $\enpr{b_1}{b_M} = \borders{i, k - 1}$ 
and $\enpr{c_1}{c_{K}} = \borders{j, k - 1}$. Let also $b_{M + 1} = c_{K + 1} = k$.
Let $x$ and $y$ be such that $b_x = a = c_y$.

Theorem~\ref{thr:update} states that there is an integer $N$ such that $\enpr{b_1}{b_N} = \borders{i, k}$
and an integer $L$ such that $\enpr{c_1}{c_L} = \borders{j, k}$.
Since $a \in \borders{i}$, we must have $N \geq x$ and similarly
$L \geq y$. 
Note that, by the induction assumption, we have $\enpr{b_x}{b_{M + 1}} = \enpr{c_y}{c_{K + 1}}$.
This implies that \textsc{update} will process exactly the same input, and deletes
exactly the same number of entries, that is, implies that $M - y \leq N - x$. This proves
the induction step.
\end{proof}

This theorem allows us to group border lists into a tree.  Let $D$ be a
sequence and let $C$ be a set of indices. We define a border tree $T =
\btree{D, C}$ as follows: The non-root nodes of the tree consists of the borders
from $\borders{c}$ for each $c \in C$, that is,
\[
	V(T) = \set{b \mid b \in \borders{c} \text{ for some } c \in C}\quad.
\]
There is an edge from a node $m$ to a node $n$ if and only if there is $c \in
C$ such that $\enpr{b_1}{b_M} = \borders{c}$, $n = b_{j}$ and $m = b_{j + 1}$.
Note that this is well-defined
since Theorem~\ref{thr:tree} states that if we have a node $n$, essentially a
border, shared by several border lists, then each border list will have the
exactly same next border, which is represented by the parent of $n$. 
Finally, the last border from each $\borders{c}$ is a child of a root, which we will denote by $r$.
Note that, for each $c \in C$, a path from $c$ to $r$ in $T$ is equal to
$\enpr{c = b_1}{b_M, r}$, where $\enpr{b_1}{b_M} = \borders{c}$.

Given a node $a$ in $\btree{D, C}$ we write $\children{a}$ to be the child
nodes of $a$. We assume that $\btree{D, C}$ is constructed so that the children
are ordered from smallest to largest. In order to be able to modify the tree quickly,
we store the tree structure as follows. Each node can have 3 pointers at most:
a pointer to a right sibling, a pointer to a left sibling or to the parent, if
there is no left sibling, and a pointer to the first child, see
Figure~\ref{fig:bt:a} as an example.

\begin{figure}
\tikzstyle{nextedge} = [yafcolor2]
\tikzstyle{prevedge} = [yafcolor5]
\tikzstyle{childedge} = [yafcolor3]

\begin{center}
\subfigure[current border tree\label{fig:bt:a}]{
\begin{minipage}[b]{3.5cm}
\begin{center}
\begin{tikzpicture}[>=latex, thick]

\node[inner sep = 1pt] (cl) {\scriptsize child};
\node[left=15pt of cl.west] {} edge[childedge, ->] (cl.west);
\node[inner sep = 1pt, below = 1pt of cl.south west, anchor = north west] (nl) {\scriptsize next};
\node[left=15pt of nl.west] {} edge[nextedge, ->] (nl.west);
\node[inner sep = 1pt, below = 1pt of nl.south west, anchor = north west] (pl) {\scriptsize prev};
\node[left=15pt of pl.west] {} edge[prevedge, ->] (pl.west);

\path[edge from parent/.style={}, level distance = 7mm, sibling distance = 8mm, inner sep = 1pt]
node[anchor  = north east] (root) at (-1, -0.5) {r}
	child {node (n7) {7}
		child {node (n4) {4}
			child {node (n1) {1}}
			child {node (n3) {3}}
		}
	}
	child {node (n9) {9}
		child {node (n8) {8}}
	}
	child {node (n10) {10}
	};

\draw[->, bend left = 15, nextedge] (n7) edge (n9);
\draw[->, bend left = 15, nextedge] (n9) edge (n10);
\draw[->, bend left = 15, nextedge] (n1) edge (n3);

\draw[->, bend left = 15, prevedge] (n9) edge (n7);
\draw[->, bend left = 15, prevedge] (n10) edge (n9);
\draw[->, bend left = 15, prevedge] (n3) edge (n1);
\draw[->, bend left = 15, prevedge] (n1) edge (n4);
\draw[->, bend left = 15, prevedge] (n4) edge (n7);
\draw[->, bend left = 15, prevedge] (n7) edge (root);
\draw[->, bend left = 15, prevedge] (n8) edge (n9);

\draw[->, bend left = 15, childedge] (n4) edge (n1);
\draw[->, bend left = 15, childedge] (n7) edge (n4);
\draw[->, bend left = 15, childedge] (root) edge (n7);
\draw[->, bend left = 15, childedge] (n9) edge (n8);
\end{tikzpicture}
\end{center}
\end{minipage}}
\subfigure[adding 11\label{fig:bt:b}]{
\begin{minipage}[b]{3.5cm}
\begin{center}
\begin{tikzpicture}[>=latex, thick]
\path[edge from parent/.style={}, level distance = 7mm, sibling distance = 8mm, inner sep = 1pt]
node (root) {r}
	child {node (n11) {11}
	child {node (n7) {7}
		child {node (n4) {4}
			child {node (n1) {1}}
			child {node (n3) {3}}
		}
	}
	child {node (n9) {9}
		child {node (n8) {8}}
	}
	child {node (n10) {10}
	}
	};

\draw[->, bend left = 15, nextedge] (n7) edge (n9);
\draw[->, bend left = 15, nextedge] (n9) edge (n10);
\draw[->, bend left = 15, nextedge] (n1) edge (n3);

\draw[->, bend left = 15, prevedge] (n9) edge (n7);
\draw[->, bend left = 15, prevedge] (n10) edge (n9);
\draw[->, bend left = 15, prevedge] (n3) edge (n1);
\draw[->, bend left = 15, prevedge] (n1) edge (n4);
\draw[->, bend left = 15, prevedge] (n4) edge (n7);
\draw[->, bend left = 15, prevedge] (n7) edge (n11);
\draw[->, bend left = 15, prevedge] (n8) edge (n9);
\draw[->, bend left = 15, prevedge] (n11) edge (root);

\draw[->, bend left = 15, childedge] (n4) edge (n1);
\draw[->, bend left = 15, childedge] (n7) edge (n4);
\draw[->, bend left = 15, childedge] (n11) edge (n7);
\draw[->, bend left = 15, childedge] (n9) edge (n8);
\draw[->, bend left = 15, childedge] (root) edge (n11);
\end{tikzpicture}
\end{center}
\end{minipage}}
\subfigure[splitting 7 from 11\label{fig:bt:c}]{
\begin{minipage}[b]{3.5cm}
\begin{center}
\begin{tikzpicture}[>=latex, thick]
\path[edge from parent/.style={}, level distance = 7mm, sibling distance = 8mm, inner sep = 1pt]
node (root) {r}
	child {node (n7) {7}
		child {node (n4) {4}
			child {node (n1) {1}}
			child {node (n3) {3}}
		}
	}
	child {node (n11) {11}
		child[missing]
		child {node (n9) {9}
			child {node (n8) {8}}
		}
		child {node (n10) {10}}
	};

\draw[->, bend left = 15, nextedge] (n7) edge (n11);
\draw[->, bend left = 15, nextedge] (n9) edge (n10);
\draw[->, bend left = 15, nextedge] (n1) edge (n3);

\draw[->, bend left = 15, prevedge] (n11) edge (n7);
\draw[->, bend left = 15, prevedge] (n10) edge (n9);
\draw[->, bend left = 15, prevedge] (n3) edge (n1);
\draw[->, bend left = 15, prevedge] (n1) edge (n4);
\draw[->, bend left = 15, prevedge] (n4) edge (n7);
\draw[->, bend left = 15, prevedge] (n9) edge (n11);
\draw[->, bend left = 15, prevedge] (n8) edge (n9);
\draw[->, bend left = 15, prevedge] (n7) edge (root);

\draw[->, bend left = 15, childedge] (n4) edge (n1);
\draw[->, bend left = 15, childedge] (n7) edge (n4);
\draw[->, bend left = 15, childedge] (n11) edge (n9);
\draw[->, bend left = 15, childedge] (n9) edge (n8);
\draw[->, bend left = 15, childedge] (root) edge (n7);
\end{tikzpicture}
\end{center}
\end{minipage}}

\subfigure[splitting 9 from 11\label{fig:bt:d}]{
\begin{minipage}[b]{3.5cm}
\begin{center}
\begin{tikzpicture}[>=latex, thick]
\path[edge from parent/.style={}, level distance = 7mm, sibling distance = 8mm, inner sep = 1pt]
node (root) {r}
	child {node (n7) {7}
		child {node (n4) {4}
			child {node (n1) {1}}
			child {node (n3) {3}}
		}
	}
	child {node (n9) {9}
		child {node (n8) {8}}
	}
	child {node (n11) {11}
		child {node (n10) {10}}
	};

\draw[->, bend left = 15, nextedge] (n7) edge (n9);
\draw[->, bend left = 15, nextedge] (n9) edge (n11);
\draw[->, bend left = 15, nextedge] (n1) edge (n3);

\draw[->, bend left = 15, prevedge] (n9) edge (n7);
\draw[->, bend left = 15, prevedge] (n10) edge (n11);
\draw[->, bend left = 15, prevedge] (n3) edge (n1);
\draw[->, bend left = 15, prevedge] (n1) edge (n4);
\draw[->, bend left = 15, prevedge] (n4) edge (n7);
\draw[->, bend left = 15, prevedge] (n11) edge (n9);
\draw[->, bend left = 15, prevedge] (n8) edge (n9);
\draw[->, bend left = 15, prevedge] (n7) edge (root);

\draw[->, bend left = 15, childedge] (n4) edge (n1);
\draw[->, bend left = 15, childedge] (n7) edge (n4);
\draw[->, bend left = 15, childedge] (n11) edge (n10);
\draw[->, bend left = 15, childedge] (n9) edge (n8);
\draw[->, bend left = 15, childedge] (root) edge (n7);
\end{tikzpicture}
\end{center}
\end{minipage}}
\subfigure[splitting 8 from 9\label{fig:bt:e}]{
\begin{minipage}[b]{3.5cm}
\begin{center}
\begin{tikzpicture}[>=latex, thick]
\path[edge from parent/.style={}, level distance = 7mm, sibling distance = 8mm, inner sep = 1pt]
node (root) {r}
	child[missing]
	child {node (n7) {7}
		child {node (n4) {4}
			child[missing]
			child {node (n1) {1}}
			child {node (n3) {3}}
		}
	}
	child {node (n8) {8}}
	child {node (n9) {9}}
	child {node (n11) {11}
		child {node (n10) {10}}
	};

\draw[->, bend left = 15, nextedge] (n7) edge (n8);
\draw[->, bend left = 15, nextedge] (n9) edge (n11);
\draw[->, bend left = 15, nextedge] (n1) edge (n3);
\draw[->, bend left = 15, nextedge] (n9) edge (n8);

\draw[->, bend left = 15, prevedge] (n8) edge (n7);
\draw[->, bend left = 15, prevedge] (n10) edge (n11);
\draw[->, bend left = 15, prevedge] (n3) edge (n1);
\draw[->, bend left = 15, prevedge] (n1) edge (n4);
\draw[->, bend left = 15, prevedge] (n4) edge (n7);
\draw[->, bend left = 15, prevedge] (n11) edge (n9);
\draw[->, bend left = 15, prevedge] (n8) edge (n9);
\draw[->, bend left = 15, prevedge] (n7) edge (root);

\draw[->, bend left = 15, childedge] (n4) edge (n1);
\draw[->, bend left = 15, childedge] (n7) edge (n4);
\draw[->, bend left = 15, childedge] (n11) edge (n10);
\draw[->, bend left = 15, childedge] (root) edge (n7);
\end{tikzpicture}
\end{center}
\end{minipage}}
\subfigure[new border tree\label{fig:bt:f}]{
\begin{minipage}[b]{3.5cm}
\begin{center}
\begin{tikzpicture}[>=latex, thick]
\path[edge from parent/.style={}, level distance = 7mm, sibling distance = 8mm, inner sep = 1pt]
node (root) {r}
	child {node (n7) {7}
		child {node (n4) {4}
			child {node (n1) {1}}
			child {node (n3) {3}}
		}
	}
	child {node (n8) {8}}
	child {node (n11) {11}
		child {node (n10) {10}
	}
	};
\draw[->, bend left = 15, nextedge] (n7) edge (n8);
\draw[->, bend left = 15, nextedge] (n8) edge (n11);
\draw[->, bend left = 15, nextedge] (n1) edge (n3);

\draw[->, bend left = 15, prevedge] (n8) edge (n7);
\draw[->, bend left = 15, prevedge] (n11) edge (n8);
\draw[->, bend left = 15, prevedge] (n3) edge (n1);
\draw[->, bend left = 15, prevedge] (n1) edge (n4);
\draw[->, bend left = 15, prevedge] (n4) edge (n7);
\draw[->, bend left = 15, prevedge] (n7) edge (root);
\draw[->, bend left = 15, prevedge] (n10) edge (n11);

\draw[->, bend left = 15, childedge] (n4) edge (n1);
\draw[->, bend left = 15, childedge] (n7) edge (n4);
\draw[->, bend left = 15, childedge] (root) edge (n7);
\draw[->, bend left = 15, childedge] (n11) edge (n10);
\end{tikzpicture}
\end{center}
\end{minipage}}

\end{center}
\caption{Border trees related to Example~\ref{ex:border:3}, demonstrating how these trees are updated. 
Figure~\ref{fig:bt:a} is given as input to \updatetree. First, \updatetree adds a new node to the tree, shown in Figure~\ref{fig:bt:b}, then proceeds to prune
obsolete borders, resulting in a new border tree, given in Figure~\ref{fig:bt:f}}
\end{figure}
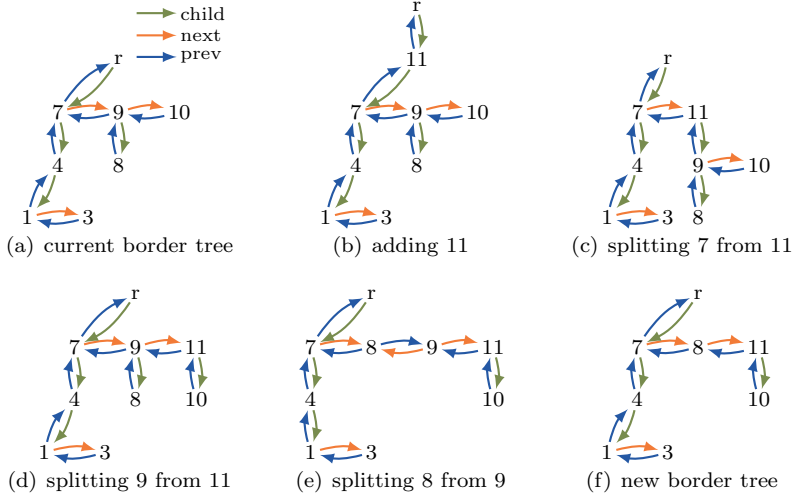

Our next step is to show how to extract the maximal average, and by doing so
compute the right interval. In order to do so we need the following results.

\begin{theorem}
\label{thr:share}
Let $D$ be a sequence  and let $1 \leq i \leq j \leq \abs{D}$ be two indices.
If $a \in \borders{i}$ and $a \geq j$, then $a \in \borders{j}$.
\end{theorem}

\begin{proof}
Assume that $a \notin \borders{j}$. Then Theorem~\ref{thr:bordfreq}
implies that there are $j \leq x < a \leq y \leq \abs{D}$ such that $\freq{x, a - 1}
\geq \freq{a, y}$.  Since $i \leq x$, Theorem~\ref{thr:bordfreq} immediately
implies $a \notin \borders{i}$.
\end{proof}

\begin{corollary}
\label{corr:sharemax}
Let $D$ be a sequence  and let $1 \leq i \leq j \leq \abs{D}$ be two indices.
If $a = \max \borders{i}$ and $a \geq j$, then $a = \max \borders{j}$.
\end{corollary}

\begin{proof}
Theorem~\ref{thr:share} implies $a \in \borders{j}$. Since both border
lists share $a$ they also share any border larger than $a$. If both $b  \in \borders{j}$ and $b > a$,
then Theorem~\ref{thr:tree} implies $b \in \borders{i}$, which is a contradiction.
Consequently, $a = \max \borders{j}$.
\end{proof}

\begin{corollary}
\label{corr:treemax}
Let $D$ be a sequence and let $C$ be a set of indices. Let $\btree{D, C}$ be a
border tree and let $r$ be its root. Select $c \in C$ and let $a \in \children{r}$ be the smallest
index such that $c \leq a$. Then $\freq{a, \abs{D}} \geq \freq{b, \abs{D}}$ for any $b \geq c$.
\end{corollary}

\begin{proof}
Let $b = \max \borders{c}$ be the maximal border.
Theorem~\ref{thr:bordmax} states that we need to prove $a = b$.
We see immediately
that $a \leq b$.  Let $d$ be such that $a = \max \borders{d}$. If $d \leq c$, then, since $c \leq a$,
Corollary~\ref{corr:sharemax} implies $a = \max \borders{c} = b$. On the other hand, if $d > c$, then since $d \leq a \leq b$, Corollary~\ref{corr:sharemax} implies
$b = \max \borders{d} = a$.
\end{proof}

Corollary~\ref{corr:treemax} gives a way to find the maximal average.  Given
$\btree{D, C}$ and $c \in C$, we look for the smallest
child of root, say $a$, such that $a \geq c$.

Our next step is to update a border tree from $T = \btree{D[1, i -
1], C}$ to $\btree{D[1, i], C}$, an update step similar to
Algorithm~\ref{alg:border}.  We start by first adding a node $i$
between a root and its children. This 
corresponds to the first two lines in Algorithm~\ref{alg:border}. After this
we modify the tree such that Theorem~\ref{thr:update} holds for every
path from $c \in C$ to the root. In Algorithm~\ref{alg:border} we simply
deleted indices that were no longer borders.  However, since a single node $n$ can be
shared by several border lists we cannot just delete it, since it might be the
case that it is still used by another border list.  Instead, we reattach
children of $n$ violating Theorem~\ref{thr:update} to the root; 
effectively removing $n$ from the border lists in which $n$ is no longer a
border. We give the pseudo-code in Algorithm~\ref{alg:tree}.

\begin{algorithm}[htb!]
\caption{$\textsc{UpdateTree}(T, C, D, i)$}
\label{alg:tree}
\Input{A tree $T = \btree{D[1, i - 1], C}$, a set of candidates $C$, a sequence $D$, an index $i$}
\Output{border tree $\btree{D[1, i], C}$}
	add node $i$ between the root and its children\;
	$a \define i$\;
	\While {$a$ exists} {
		$n \define $ next sibling of $a$\;

		\If {$a$ is a leaf} {
		 	\lIf {$a \notin C$} {
				delete $a$ from $T$
			}
		}
		\Else {
			$b \define $ first child of $a$\;
			\If {$\freq{b, i} \geq \freq{a, i}$} {
			\nllabel{alg:tree:test}
				detach $b$ from $a$\;
				attach $b$ to the root left to $a$\;
				$n \define b$\;
			}
		}
		$a \define n$\;
	}
	\Return $T$\;
\end{algorithm}

\begin{example}
\label{ex:border:3}
Let us continue Examples~\ref{ex:border:1}--\ref{ex:border:2}.  Assume that we
have a sequence given in Example~\ref{ex:border:1} and that we have $C =
\set{1, 3, 8, 10}$. Based on borders given in Example~\ref{ex:border:2}, the
border tree is given in Figure~\ref{fig:bt:a}.  Assume that we see a new data
point, $D_{11} = 1$.  We have $\borders{1, 11} = (1, 4, 7)$, $\borders{3, 11} =
(3, 4, 7)$, $\borders{8, 11} = (8)$, and $\borders{10, 11} = (10, 11)$.

We begin updating the tree by first adding node $11$ between the root and its
children, see Figure~\ref{fig:bt:b}. We continue by checking the first child of
$11$: node $7$, and reattach it to $r$, see Figure~\ref{fig:bt:c}. After this,
we check the first child of $7$, node $4$ and leave it unmodified. We continue
by reattaching $9$ to the root, see Figure~\ref{fig:bt:d}, and similarly node
$8$, see Figure~\ref{fig:bt:e}. Since node $9$ is now a leaf and $9 \notin C$,
we can delete it. Finally, we leave $10$ attached to $11$. The final tree,
which corresponds to the correct border tree, is given in Figure~\ref{fig:bt:f}.
\end{example}

\begin{theorem}
\label{thr:btreecorrect}
Let $T = \btree{D[1, i - 1], C}$.
Algorithm $\textsc{UpdateTree}(T, C, D, i)$ outputs $\btree{D[1, i], C}$.
\end{theorem}

See Appendix for the proof.

In addition to \updatetree, we need a routine for updating the 
tree when an index $c$ is deleted from $C$. This is needed when \segment deletes
a candidate for the optimal segmentation. In order to update we simply
check whether $c$ is a leaf, if it is, then we delete it, and recursively test
the parent of $c$.

Finally, let us address memory and time complexity of a border tree.  First of
all, we have $\abs{D}$ nodes at maximum, hence we need $O(\abs{D})$ memory.
Let $L$ be the maximum number of $\abs{C}$. Let $K_i$ be the number of nodes
removed during $\updatetree(T, D, C, i)$.  If we do not modify the tree during
the while-loop, then we execute the while-loop only once, since there is
only child of $r$, namely $i$.  Note that by the end of each $\updatetree(T, D,
C, i)$, root $r$ can have at most $L$ children. This means that at maximum we
have done $L + K_i$ reattachments.  Each reattachment increases the while-loop
executions by 2: we need to check the child attached to the root and we need to
check whether the parent has more children that need to be reattached. Hence,
the while-loop is executed at most $2(L + K_i) + 1$ times during
$\updatetree(T, D, C, i)$. Thus total time complexity is $O(\abs{D}L + \sum_{i
= 1}^{\abs{D}}K_i)$.  Note that once a node is deleted it will not be introduced again.
Hence, $\sum_{i = 1}^{\abs{D}}K_i \leq \abs{D}$.  This gives us a total execution
time of $O(\abs{D}L)$.

%% file: experiments.tex
\section{Experiments}
\label{sec:experiments}

In this section we empirically evaluate our approach on synthetic and real-world
datasets.\!\footnote{The implementation of the algorithm is given at \url{http://adrem.ua.ac.be/segmentation}}

\paragraph{Synthetic data}
Our main contribution to the paper is the speedup of the dynamic program for
finding the optimal segmentation when using one-dimensional log-linear models.  We
measure the efficiency by the total number of comparisons needed in
Line~\ref{alg:fast:inner} of Algorithm~\ref{alg:fast}.  We define a
\emph{performance ratio} by normalizing this number by the number of
comparisons that we would have made if we would not use any pruning. This ensures
that the ratio is between $0$ and $1$, smaller values indicating faster
performance.  Note that if we do not use any pruning, the total
number of comparisons is $O(K\abs{D}^2)$.

We begin by generating sequences of random samples drawn from the Gaussian
distribution with $0$ mean and $1$ variance. We generated 11 sequences of
lengths $2^k$ for $k = 10, \ldots, 20$ and computed the performance ratio of
our segmentation using $4$ segments of Gaussian distributions (as given in
Example~\ref{ex:gaussian}). From results given in Figure~\ref{fig:ind:size} we
see that we obtain speedups of 1 order of  magnitude for the smallest data, up to 3 orders of magnitude
for longer data: the ratio for the largest sequence is $0.0007$.  Note that the ratios
become smaller as the sequence becomes larger. The reason is that when considering
longer segments, it becomes more likely that we can delete 
candidates, making the algorithm relatively faster.  The absolute
computation time grows with the length of a sequence, $11ms$, $1.3s$, and $20$
minutes for sequences of length $2^{10}$, $2^{15}$, and $2^{20}$, respectively.

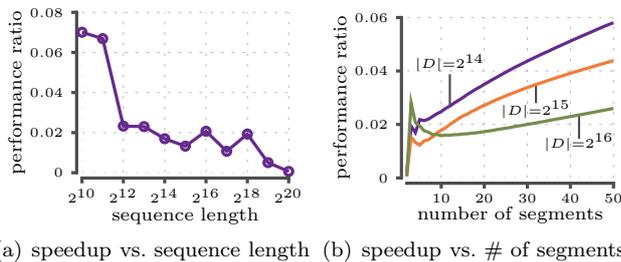
\begin{figure}[htb!]
\begin{center}
\subfigure[speedup vs. sequence length\label{fig:ind:size}]{
\begin{tikzpicture}
\begin{axis}[xlabel=sequence length, ylabel= {performance ratio},
	width = 4.3cm,
	xmax = 20,
	xmin = 10,
	ymin = 0,
	ymax = 0.08,
	scaled y ticks = false,
	cycle list name=yaf,
	yticklabel style={/pgf/number format/fixed},
	xticklabel = {$\scriptstyle 2^{\pgfmathprintnumber{\tick}}$},
	yticklabel = {$\scriptstyle\pgfmathprintnumber{\tick}$}
	]
\addplot table[x expr = {\lineno + 10}, y index = 2, header = false] {indratiosize.dat};
\pgfplotsextra{\yafdrawaxis{10}{20}{0}{0.08}}
\end{axis}
\end{tikzpicture}}
\subfigure[speedup vs. \# of segments\label{fig:ind:seg}]{
\begin{tikzpicture}
\begin{axis}[xlabel=number of segments, ylabel= {performance ratio},
	width = 4.3cm,
	xmax = 50,
	xmin = 2,
	ymin = 0,
	ymax = 0.06,
	no markers,
	scaled y ticks = false,
	xtick = {10, 20, 30, 40, 50},
	cycle list name=yaf,
	yticklabel style={/pgf/number format/fixed},
	xticklabel = {$\scriptstyle\pgfmathprintnumber{\tick}$},
	yticklabel = {$\scriptstyle\pgfmathprintnumber{\tick}$}
	]
\node[coordinate, pin={[pin edge = {yafaxiscolor, thick, shorten > = -2pt}, pin distance = 10pt]above:$\scriptscriptstyle \abs{D} = 2^{14}$}] at (axis cs:12, 0.026882) {}; 
\node[coordinate, pin={[pin edge = {yafaxiscolor, thick, shorten > = -4pt}, pin distance = 3pt]below:$\scriptscriptstyle \abs{D} = 2^{15}$}] at (axis cs:32, 0.034937) {}; 
\node[coordinate, pin={[pin edge = {yafaxiscolor, thick, shorten > = -5pt}, pin distance = 4pt]below:$\scriptscriptstyle \abs{D} = 2^{16}$}] at (axis cs:42, 0.023550) {}; 

\addplot table[x expr = {\lineno + 2}, y index = 0, header = false] {indratioseg.dat};
\addplot table[x expr = {\lineno + 2}, y index = 1, header = false] {indratioseg.dat};
\addplot table[x expr = {\lineno + 2}, y index = 2, header = false] {indratioseg.dat};

\pgfplotsextra{\yafdrawaxis{2}{50}{0}{0.06}}
\end{axis}
\end{tikzpicture}}
\end{center}
\caption{Performance ratio, total number of score comparisons 
(see Algorithm~\ref{alg:fast}, Line~\ref{alg:fast:inner}), normalized between $0$ and $1$,
as a function of sequence length~\ref{fig:ind:size}, using $4$ segments, and as a function of number of segments~\ref{fig:ind:seg}. Smaller values are better}
\end{figure}

Our second experiment is to study the performance ratio as a function of
segments. We sampled 3 sequences from a Gaussian distribution, with $0$ mean and
$1$ variance, of sizes $2^{14}$, $2^{15}$, $2^{16}$. For each sequence we
computed segmentations up to $50$ segments. From the results given in
Figure~\ref{fig:ind:seg} we see that the performance ratio becomes worse as we
increase the number of segments. The reason is that when segments 
become shorter, consequently, the right intervals are more compact and have
less chance of being intersected with the left interval. Nevertheless, we get
$0.06$, $0.04$, and $0.02$ for performance ratios for our sequences when using
$50$ segments. The peak at $3$ segments suggest that discovering segmentation
with $3$ segments is particularly expensive. To see why this is happening, first
note that the first segment always starts from the beginning. This implies that
when looking for a segmentation with $2$ segments for a sequence $D[1, i]$, the
second segment will be typically either really short or really long as its mean
needs to differ from the mean of the first segment. If the second segment is
short, it will have an abnormal right interval, consequently, the interval has
a smaller chance of overlapping with the left interval of the next segment.

\begin{figure}
\begin{center}
\subfigure[sequence\label{fig:step:a}]{
\begin{tikzpicture}
\begin{axis}[xlabel=index, 
	width = 4cm,
	xmax = 4000,
	ymin = -7,
	ymax = 7,
	xtick = {0, 1000, 2000, 3000, 4000},
	x tick label style = {/pgf/number format/set thousands separator = {\,}},
	cycle list name=yaf,
	no markers,
	xticklabel = {$\scriptstyle\pgfmathprintnumber{\tick}$},
	yticklabel = {$\scriptstyle\pgfmathprintnumber{\tick}$}
	]
\addplot[yafcolor2] table[x expr = {\lineno * 40}, y index = 0, header = false] {step.life};
\pgfplotsextra{\yafdrawaxis{0}{4000}{-7}{7}}
\end{axis}
\end{tikzpicture}}
\subfigure[candidate lifetime, $K = 2$]{
\begin{tikzpicture}
\begin{axis}[xlabel=index, ylabel= {lifetime},
	width = 4cm,
	xmax = 4000,
	ymin = 0,
	ymax = 200,
	xmin = 0,
	ytick = {0, 50, 100, 150, 200},
	xtick = {0, 1000, 2000, 3000, 4000},
	x tick label style = {/pgf/number format/set thousands separator = {\,}},
	cycle list name=yaf,
	no markers,
	xticklabel = {$\scriptstyle\pgfmathprintnumber{\tick}$},
	yticklabel = {$\scriptstyle\pgfmathprintnumber{\tick}$}
	]
\begin{scope}
\path[clip] (axis cs: -40, -2.5) rectangle (axis cs: 4000, 200); 
\addplot+[fill, yafcolor1!30] table[x expr = {\lineno * 40}, y index = 1, header = false] {step.life} \closedcycle;
\addplot+[yafcolor1] table[x expr = {\lineno * 40}, y index = 1, header = false] {step.life};
\end{scope}
\pgfplotsextra{\yafdrawaxis{0}{4000}{0}{200}}
\end{axis}
\end{tikzpicture}}

\subfigure[candidate lifetime, $K = 3$]{
\begin{tikzpicture}
\begin{axis}[xlabel=index, ylabel= {lifetime},
	width = 4cm,
	xmax = 4000,
	ymin = 0,
	ymax = 500,
	ytick = {0, 100, 200, 300, 400, 500},
	xtick = {0, 1000, 2000, 3000, 4000},
	x tick label style = {/pgf/number format/set thousands separator = {\,}},
	cycle list name=yaf,
	no markers,
	xticklabel = {$\scriptstyle\pgfmathprintnumber{\tick}$},
	yticklabel = {$\scriptstyle\pgfmathprintnumber{\tick}$}
	]
\begin{scope}
\path[clip] (axis cs: -40, -5) rectangle (axis cs: 4000, 500); 
\addplot+[fill, yafcolor1!30] table[x expr = {\lineno * 40}, y index = 2, header = false] {step.life} \closedcycle;
\addplot+[yafcolor1] table[x expr = {\lineno * 40}, y index = 2, header = false] {step.life};
\end{scope}
\clip; 
\pgfplotsextra{\yafdrawaxis{0}{4000}{0}{500}}
\end{axis}
\end{tikzpicture}}
\subfigure[candidate lifetime, $K = 4$]{
\begin{tikzpicture}
\begin{axis}[xlabel=index, ylabel= {lifetime},
	width = 4cm,
	xmax = 4000,
	ymin = 0,
	ymax = 900,
	ytick = {0, 300, 600, 900},
	xtick = {0, 1000, 2000, 3000, 4000},
	x tick label style = {/pgf/number format/set thousands separator = {\,}},
	y tick label style = {/pgf/number format/set thousands separator = {\,}},
	cycle list name=yaf,
	no markers,
	xticklabel = {$\scriptstyle\pgfmathprintnumber{\tick}$},
	yticklabel = {$\scriptstyle\pgfmathprintnumber{\tick}$}
	]
\begin{scope}
\path[clip] (axis cs: -40, -12) rectangle (axis cs: 4000, 900); 
\addplot+[fill, yafcolor1!30] table[x expr = {\lineno * 40}, y index = 3, header = false] {step.life} \closedcycle;
\addplot+[yafcolor1] table[x expr = {\lineno * 40}, y index = 3, header = false] {step.life};
\end{scope}
\pgfplotsextra{\yafdrawaxis{0}{4000}{0}{900}}
\end{axis}
\end{tikzpicture}}
\end{center}
\caption{Sequence of $4$ Gaussian segments and candidate lifetimes, how many
iterations is needed for a candidate to be deleted, when computing a segmentation
with $K$ segments from a segmentations of $K - 1$ segments,
where $K = 2, 3, 4$. Smaller values imply lower computational burden}
\label{fig:step}
\end{figure}
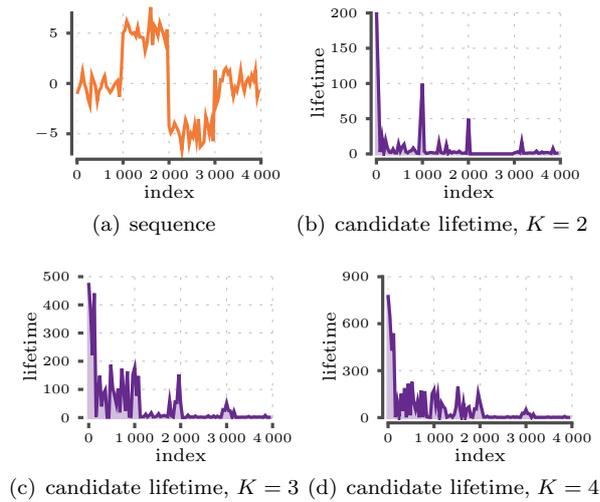

Our next step is to study how candidates for segments are distributed.  A
candidate $c$ is added to $C$ on Line~\ref{alg:fast:add} and deleted from $C$
on Line~\ref{alg:fast:delete} in \segment. The candidate is added when the
counter $i$ is equal to $c$ and let us assume that it is deleted when the
counter is equal to $j$. If $c$ is not deleted, after the for-loop in \segment,
we simply set $j = \abs{D} + 1$. We define a \emph{lifetime} of a candidate $c$
to be $j - c$, that is, a candidate lifetime is how often it has been used in
the maximization step on Line~\ref{alg:fast:inner}. The smaller the value, the
less computational burden a candidate is producing. In the worst case, that is,
without any pruning, the lifetime for a candidate $c$ is equal to
$\abs{D} + 1 - c$.

To study candidate lifetimes we generate a sequence of $4\,000$ samples,
consisting of $4$ segments of Gaussian distribution with $0$, $5$, $-5$, and
$0$ means, respectively, and variance of $1$ (see Figure~\ref{fig:step:a}).  We
computed segmentations up to $4$ segments and present the lifetimes in
Figure~\ref{fig:step}.\!\footnote{For clarity sake, figures show average
lifetimes of bins containing $40$ points} We see that there are four
major spikes in lifetimes, at the beginning of the sequence and around each
change point.  Let us consider a spike at $2\,000$ for $K = 4$. A candidate on the
left side of the spike has a longer lifetime because the left interval of the
next segment is shifted and it is less likely that it will intersect with the
right interval.  On the other hand, a candidate on the right side of the spike
has a longer lifetime because the segment is short and the right interval has a
higher chance of being abnormal. The same rationale applies to spike at the
beginning of the sequence. The spikes grow with increasing number of segments,
nevertheless they are shallow, implying that we have significant speedup.
In fact, the performance ratios are $0.004$, $0.01$, $0.02$ for segmentations
with $K = 2, 3, 4$ segments, respectively.

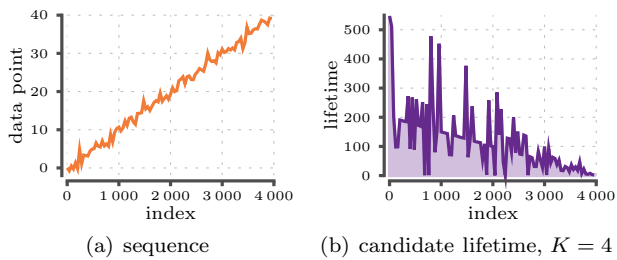
\begin{figure}[htb!]
\begin{center}
\subfigure[sequence\label{fig:slope:a}]{
\begin{tikzpicture}
\begin{axis}[xlabel=index, ylabel= {data point},
	width = 4.3cm,
	xmax = 4000,
	ymin = -2,
	ymax = 40,
	xtick = {0, 1000, 2000, 3000, 4000},
	x tick label style = {/pgf/number format/set thousands separator = {\,}},
	cycle list name=yaf,
	no markers,
	xticklabel = {$\scriptstyle\pgfmathprintnumber{\tick}$},
	yticklabel = {$\scriptstyle\pgfmathprintnumber{\tick}$}
	]
\addplot[yafcolor2] table[x expr = {\lineno * 40}, y index = 0, header = false] {slope.life};
\pgfplotsextra{\yafdrawaxis{0}{4000}{-2}{40}}
\end{axis}
\end{tikzpicture}}
\subfigure[candidate lifetime, $K = 4$\label{fig:slope:b}]{
\begin{tikzpicture}
\begin{axis}[xlabel=index, ylabel= {lifetime},
	width = 4.3cm,
	xmax = 4000,
	ymin = 0,
	ymax = 550,
	ytick = {0, 100, 200, 300, 400, 500},
	xtick = {0, 1000, 2000, 3000, 4000},
	x tick label style = {/pgf/number format/set thousands separator = {\,}},
	cycle list name=yaf,
	no markers,
	xticklabel = {$\scriptstyle\pgfmathprintnumber{\tick}$},
	yticklabel = {$\scriptstyle\pgfmathprintnumber{\tick}$}
	]
\addplot+[fill, yafcolor1!30] table[x expr = {\lineno * 40}, y index = 2, header = false] {slope.life} \closedcycle;
\addplot+[yafcolor1] table[x expr = {\lineno * 40}, y index = 2, header = false] {slope.life};
\pgfplotsextra{\yafdrawaxis{0}{4000}{0}{550}}
\end{axis}
\end{tikzpicture}}
\end{center}

\caption{Sequence sampled from a Gaussian distribution with a slowly increasing mean and candidate lifetimes, how many
iterations is needed for a candidate to be deleted, when computing a segmentation
with $K$ segments from a segmentations of $K - 1$ segments, for segmentation with $K$ segments}
\end{figure}

Finally, let us demonstrate the limitations of our approach.  We generate a
sequence of $4\,000$ samples, where a sample $i$ is generated from a Gaussian
distribution with a mean of $i / 100$ and variance of $1$, see
Figure~\ref{fig:slope:a}. The performance ratio of segmentation with $4$
segments is $0.06$, the lifetimes are given in Figure~\ref{fig:slope:b}. While
we see a good performance for this data, when we increase the slope (or
equivalently, lower variance) the performance ratio becomes worse.  The
worst case scenario is a genuinely monotonically increasing (or
decreasing) sequence, that is, $D_{i + 1} > D_i$.  In such case, the left
intervals and the right intervals will never overlap and no candidate will be
pruned. We should point out that applying segmentation for a monotonic sequence
in the first place is questionable as such sequence does not fit well
the segmentation probabilistic model, and it might be beneficial to detrend the
data to obtain a better segmentation.

\paragraph{Real-world data}
We continue our experiments using real-world data sets.  We considered $3$
different datasets.\!\footnote{The datasets were obtained from
\url{http://www.cs.ucr.edu/~eamonn/discords/}} The first dataset,
\emph{Marotta}, is Space Shuttle Marotta Valve time series, consisting of 5
energize/de-energize cycles (TEK17). The second dataset, \emph{Power}, consists
of a power consumption of a Dutch research facility during the year 1997.  The
third dataset consists of two-dimensional time series extracted from videos
of an actor performing various acts with and without a replica gun. Since
this sequence is two-dimensional, we split the dimensions into 
\emph{Video1} and \emph{Video2}. The sequence lengths are given in Table~\ref{tab:real}.

\begin{table}[htb!]
\caption{Characteristics of real-world datasets and performance of the algorithm with $20$ segments. The last column is the time needed
to compute the optimal segmentation using traditional dynamic program}
\label{tab:real}
\begin{center}
\begin{tabular}{lrrrrr}
\toprule
Data & length & & performance & time (s) & baseline time (s)\\
\midrule
\emph{Marotta} & $5\,000$ &  & $0.04$ & $0.6$ & $13$\\
\emph{Power} & $35\,040$ &  & $0.03$ & $19.5$ & $600$\\
\emph{Video1} & $11\,251$ &  & $0.1$ & $6.7$ & $62$\\
\emph{Video2} & $11\,251$ &  & $0.14$ & $9.7$ & $62$\\
\bottomrule
\end{tabular}
\end{center}
\end{table}

We study the performance by computing segmentations with $20$ segments for each
data and comparing it against the traditional dynamic program, that is, without
deleting any candidates.  From the results, given in Table~\ref{tab:real}, we
see that our approach has a significant advantage over a baseline approach, for
example, with \emph{Power} dataset we find an optimal solution in 20 seconds
while the baseline approach requires 10 minutes.

Finally, let us look at some of the discovered segmentations.  In
Figure~\ref{fig:tek17} we present a segmentation of \emph{Marotta} with $11$
segments.  The segments align with high and low energy states.  Note that the
$3$rd high energy segment is more shallow than the other high energy segments.
This cycle contains an anomaly as pointed out by~\citet{keogh:05:hot-sax}
resulting in a shorter high energy segment. In Figure~\ref{fig:power} we show a
segmentation with $3$ segments of the power consumption. We can see that the
mean of the middle segment is lower than the other means, indicating a summer
season.

\begin{figure}[htb!]
\begin{center}
\begin{tikzpicture}
\begin{axis}[xlabel=index, ylabel= {data point},
	width = 9cm,
	height = 3cm,
	ymin = -2,
	ymax = 4.5,
	xtick = \empty,
	xtick = { 4433, 4160, 3404, 3150, 2330, 2165, 1390, 1151, 372, 161},
	x tick label style = {font=\scriptsize, yshift = 3pt, rotate = 45, anchor = north east},
	xmajorgrids = false,
	x tick label style = {/pgf/number format/set thousands separator = {\,}},
	cycle list name=yaf,
	xticklabel = {$\scriptstyle\pgfmathprintnumber{\tick}$},
	yticklabel = {$\scriptstyle\pgfmathprintnumber{\tick}$}
	]

\draw[yafaxiscolor, line width = 0.5pt] (axis cs: 161, -2.300000) -- (axis cs: 161, 4.500000);
\draw[yafaxiscolor, line width = 0.5pt] (axis cs: 372, -2.300000) -- (axis cs: 372, 4.500000);
\draw[yafaxiscolor, line width = 0.5pt] (axis cs: 1151, -2.300000) -- (axis cs: 1151, 4.500000);
\draw[yafaxiscolor, line width = 0.5pt] (axis cs: 1390, -2.300000) -- (axis cs: 1390, 4.500000);
\draw[yafaxiscolor, line width = 0.5pt] (axis cs: 2165, -2.300000) -- (axis cs: 2165, 4.500000);
\draw[yafaxiscolor, line width = 0.5pt] (axis cs: 2330, -2.300000) -- (axis cs: 2330, 4.500000);
\draw[yafaxiscolor, line width = 0.5pt] (axis cs: 3150, -2.300000) -- (axis cs: 3150, 4.500000);
\draw[yafaxiscolor, line width = 0.5pt] (axis cs: 3404, -2.300000) -- (axis cs: 3404, 4.500000);
\draw[yafaxiscolor, line width = 0.5pt] (axis cs: 4160, -2.300000) -- (axis cs: 4160, 4.500000);
\draw[yafaxiscolor, line width = 0.5pt] (axis cs: 4433, -2.300000) -- (axis cs: 4433, 4.500000);

\addplot[yafcolor2, no markers] table[x expr = {\lineno * 10}, y index = 0, header = false] {tek.life};
\pgfplotsextra{\yafdrawaxis{0}{5000}{-2}{4.5}}
\end{axis}
\end{tikzpicture}
\end{center}
\caption{Segmentation with 11 segments of Space Shuttle Marotta Valve time series}
\label{fig:tek17}
\end{figure}
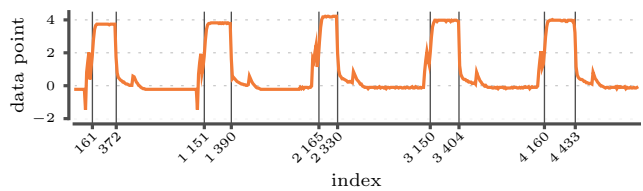

\begin{figure}[htb!]
\begin{center}
\begin{tikzpicture}
\begin{axis}[xlabel=index, ylabel= {data point},
	width = 9cm,
	height = 3cm,
	ymin = 600,
	ymax = 2000,
	xtick = {8232, 22015},
	scaled x ticks = false,
	xmajorgrids = false,
	x tick label style = {/pgf/number format/set thousands separator = {\,}},
	cycle list name=yaf,
	xticklabel = {$\scriptstyle\pgfmathprintnumber{\tick}$},
	yticklabel = {$\scriptstyle\pgfmathprintnumber{\tick}$}
	]

\draw[yafcolor1, line width = \yafaxiswidth] (axis cs: 8232, 540.000000) -- (axis cs: 8232, 2000.000000);
\draw[yafcolor1, line width = \yafaxiswidth] (axis cs: 22015, 540.000000) -- (axis cs: 22015, 2000.000000);
\draw[yafcolor1, line width = \yafaxiswidth] (axis cs: 0, 1212.914845) -- (axis cs: 8231, 1212.914845);
\draw[yafcolor1, line width = \yafaxiswidth] (axis cs: 8232, 1076.447363) -- (axis cs: 22014, 1076.447363);
\draw[yafcolor1, line width = \yafaxiswidth] (axis cs: 22015, 1172.030557) -- (axis cs: 35039, 1172.030557);

\addplot[yafcolor2, no markers] table[x expr = {\lineno * 175.2}, y index = 0, header = false] {power.life};

\pgfplotsextra{\yafdrawaxis{0}{35040}{600}{2000}}
\end{axis}
\end{tikzpicture}
\caption{Segmentation with 3 segments of the \emph{Power} dataset. The horizontal lines represent the means of the individual segments.}
\label{fig:power}
\end{center}
\end{figure}
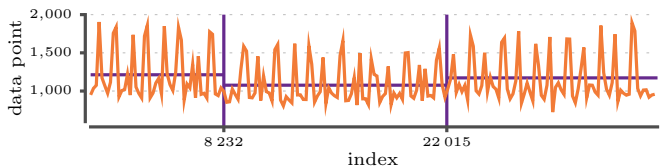

%% file: related.tex
\section{Related Work}
\label{sec:related}

Segmentation is an instance of a larger problem setting, called change point
detection, see~\citep{basseville:93:detection}, for introduction.  We can
divide the problem settings broadly into two categories: offline and online.
Although these settings have conceptually the same goal, the setup details
make it different from an algorithmic point of view. In online change point
detection (see~\citet{kifer:04:detecting}, for example) the data arrives in a
stream fashion, typically there is no budget for how many change points are
allowed, and the decision needs to be made within some time frame, whereas in
segmentation, offline change point detection, new datapoints can change
early segments. A typical goal for online change detection is to alert
the system or a user of a change, whereas in segmentation the only goal is to summarize the
sequence.

Popular approaches for segmentation are top-down approaches where we select
greedily a new change-point
(see~\citet{shatkay:96:approximate,bernaola-galvan:96:compositional,douglas:73:algorithms,lavrenko:00:mining},
for example) and bottom-up approaches where at the beginning each point is a
segment, and points are combined in a greedy fashion
(see~\citet{palpanas:04:online}, for example). A randomized heuristic 
was suggested by~\citet{himberg:01:time}, where we start from a random
segmentation and optimize the segment boundaries. 
These approaches, although fast, are heuristics and have no
theoretical guarantees of the approximation quality.  A divide-and-segment
approach, an approximation algorithm with theoretical guarantees on the
approximation quality was given by~\citet{terzi:06:efficient}.

Modifications of the original segmentation problem have been also studied.
Discovering recurrent sources is a setup where one limits the amount of distinct
means of the segments to be $H$ such that $H < K$, where $K$ is the number of
allowed segments has been suggested~\citep{gionis:03:finding}.
\citet{haiminen:04:unimodal} study unimodal sequences, where
means of the centroids (of one-dimensional sequence) are required to follow a
unimodal curve, that is, the means should only rise to some point and then only
decline afterwards.
For a survey of the segmentation algorithms, see Chapter 8 in~\citep{dzeroski:11:inductive}. 

%% file: conclusions.tex
\section{Discussion and Conclusions}
\label{sec:conclusions}
In this paper we introduced a pruning technique to speedup the dynamic program
used for solving the segmentation problem. We demonstrated on both synthetic
and real-world data that we gain a significant speedup by using our pruning
technique.

We should point out that our pruning is online, that is, the decision to delete
a candidate is based only on current and past data points.  We believe that we
can speedup the algorithm further by applying additional pruning techniques
based on future data points, such as~\citep{gedikli:10:modified}. In addition,
we conjecture that these optimizations may prove to be useful in other setups,
such as, discovering HMMs or CRFs, where dynamic programs are used in order to
optimize the model.  We leave these studies as future work.

Segmentation requires a parameter, namely the number of segments.  One approach
to remove this parameter is by using model selection techniques, such as,
BIC~\citep{schwarz:78:estimating} or MDL~\citep{grunwald:07:minimum}. We conjecture
that using these techniques not only remove the parameter but can be also used
for further speedup. 

Our algorithm is limited only to handle one-dimensional case.  However, the key
result, Theorem~\ref{thr:suff}, actually handles the multi-dimensional case.  The
reason why we limit ourselves to one-dimensional case is that we were able to
verify the sufficient conditions in Theorem~\ref{thr:suff} with relative ease. We leave
studying applying Theorem~\ref{thr:suff} more generally as future work. While we are
skeptical whether it is possible verify the conditions in Theorem~\ref{thr:suff}
exactly, we believe that it is possible to find more conservative conditions
that can be easily checked and that will imply the conditions in
Theorem~\ref{thr:suff}.

%% file: appendix.tex
\section{Proofs}
\subsection{Proof of Theorem~\ref{thr:suff}}
\label{sec:appprune}
Theorem~\ref{thr:suff} will follow from the following theorem.

\begin{theorem}
Let $D = \enpr{D_1}{D_e}$. Let $1 \leq m < e$.
Assume that $\diff{D[1, m], D[m + 1, e]}$ is a cover.
Then there exists $n > m$ such that $\score{[1, n], [n + 1, e]} > \score{[1, m], [m + 1, e]}$
or there exists $l < m$ such that $\score{[1, l], [l + 1, e]} \geq \score{[1, m], [m + 1, e]}$.
\label{thr:ascent}
\end{theorem}

In order to prove the theorem we will introduce some helpful notation. First, given
a parameter vector $s$ and $r$, we define
\[
	h(k \mid s, r) = \score{[1, k] \mid s} + \score{[k + 1, e] \mid r}\quad.
\]
Note that $h(k \mid s, r) \leq \score{[1, k], [k + 1, e]}$. We also define
\[
	 g(l, \delta \mid s, r)  = l(Z(s) - Z(r) + (s - r)^T\delta)\quad.
\]
This function is essentially the difference between two scores.
\begin{lemma}
Let $k > l$.  We have $h(k \mid s, r) - h(l \mid s, r) =g(k - l, \freq{l + 1, k} \mid s, r)$.
\end{lemma}

\begin{proof}
Note that
\[
\begin{split}
	h(k \mid s, r) & = kZ(s) + s^T\cs{k} + (e - k)Z(r) + r^T(\cs{e} - \cs{k}) \\
	               & = k(Z(s) - Z(r)) + (s - r)^T\cs{k} + eZ(r) + r^T\cs{e}\quad.
\end{split}
\]
The last two terms do not depend on $k$. This allows us to write
\[
\begin{split}
	& h(k \mid s, r) - h(l \mid s, r) = k(Z(s) - Z(r)) + (s - r)^T\cs{k} - l(Z(s) - Z(r)) - (s - r)^T\cs{l}  \\
	&\quad = (k - l)\big(Z(s) - Z(r) + (s - r)^T\frac{\cs{k} - \cs{l}}{k - l}\big) = g(k - l, \freq{l + 1, k} \mid s, r)\quad.
\end{split}
\]
This completes the proof.
\end{proof}

\begin{proof}[Proof of Theorem~\ref{thr:ascent}]
Write $y = \score{[1, m], [m + 1, e]}$ and define
\[
	x = \max_{k < m} \score{[1, k], [k + 1, e]} \quad\text{and}\quad  z = \max_{k > m} \score{[1, k], [k + 1, e]}\quad.
\]
We need to show that either $x \geq y$ or $z > y$. Assume that $z \leq y$.
Fix $\epsilon > 0$.  By definition, there exist $s$ and $r$ such that 
\[
	\score{[1, m] \mid s} + \score{[m + 1, e] \mid r} \geq y - \epsilon\quad.
\]
From now on we will write $h(k)$ to mean $h(k \mid s, r)$ and $g(k, \delta)$ to mean $g(k, \delta \mid s, r)$.
We must have $h(m) + \epsilon \geq y \geq z$ or, equivalently, $\epsilon \geq  z - h(m)$.

Since $\diff{D[1, m], D[m + 1, e]}$ is a cover, there exist integers $l$ and $n$, $0 \leq l < m < n \leq e$, such that $(\alpha - \beta)^T(s - r) \geq 0$, where
$\alpha = \freq{m + 1, n}$ and $\beta = \freq{l + 1, m}$.

Define $c = (n - m) / (m - l)$.  We now have
\[
\begin{split}
	\epsilon & \geq z - h(m) \geq h(n) - h(m) = g(n - m, \alpha) = cg(m - l, \alpha) \\
	& = cg(m - l, \beta) + c(m - l)(s - r)^T(\alpha - \beta) \geq cg(m - l, \beta) \\
	& = c(h(m) - h(l)) \geq c(h(m) - x) \geq c(y - \epsilon - x), \\
\end{split}
\]
which implies $y - x \leq \epsilon(1 + c^{-1}) \leq \epsilon(1 + e)$. Since this holds for any $\epsilon > 0$, we conclude that $y \leq x$.
This proves the theorem.
\end{proof}

\begin{proof}[Proof of Theorem~\ref{thr:suff}]
Let $P$ a segmentation
and let $I$ and $J$ be two consecutive segments such that $\diff{D[I], D[J]}$ is a cover.
We can now apply Theorem~\ref{thr:ascent} to find alternative segments $I'$ and $J'$ such that
if we define $P'$ by replacing $I$ and $J$ from $P$ with $I'$ and $J'$ then
either $\score{P' \mid D} > \score{P' \mid D}$ or $\score{P' \mid D} \geq \score{P' \mid D}$ and $I'$
ends before $I$. We repeat this until no consecutive segments constitute a cover.
This repetition ends because no segmentation will occur twice during these steps and there is a finite
number of segmentations. The reason why no segmentation occur twice is 
because either the score properly increases or the score stays the same and we move
a breakpoint to the left.
\end{proof}

\subsection{Proof of Theorem~\ref{thr:btreecorrect}}
Let $U$ be the resulting tree from $\textsc{UpdateTree}(T, C, D, i)$.
To prove the theorem we need to show that the paths of $U$ from leafs to 
the root consists of borders, there are no
nodes in $U$ outside the borders, and that children are ordered. We will prove these results
in a series of lemmata. 

\begin{lemma}
Let $T'$ be a tree after we have added a node $i$ in \updatetree.
Let $n \neq i$ be a node in $T'$ and let $m$ be its parent.
Let $c \in C$ be such that $n \in \borders{c, i - 1}$.
If $m \notin \borders{c, i}$, then $n$ will cease to be a child of $m$
during some stage of \updatetree.
\label{lem:prune}
\end{lemma}

\begin{proof}
Let $r$ be a root node of $T'$.  Consider a pre-order of
nodes of $T'$, that is, parents and earlier siblings come first.  We will prove the lemma
using induction on the pre-order.

To prove the first step, let $n$ be the first child of $i$.  If $i \notin
\borders{c, i}$, then Theorem~\ref{thr:update} implies that $\freq{n, i}
\geq \freq{i, i}$ which is exactly the test on Line~\ref{alg:tree:test}. Hence,
$n$ will be disconnected from $i$.

Let us now prove the induction step. 
Let $p$ be the parent of $m$ in $T'$. Assume that $p \neq r$.  Note that $p$ is the
border next to $m$ in $\borders{c, i - 1}$.  
Theorem~\ref{thr:update} implies that $p \notin \borders{c, i}$, hence the induction assumption
implies that $m$ and $p$ are disconnected and $m$ becomes a child of $r$ at some point.

Assume now that $n$ is not the first child of $m$ and let $q$ be the sibling
left to $n$, and let $p$ be such that $q \in \borders{p, i - 1}$.
Theorem~\ref{thr:bordmax} implies that $\freq{q, m - 1} \geq
\freq{j, m - 1}$ for any $q \leq j < m$.  Since $n > q$, we must have $\freq{q,
m - 1} \geq \freq{n, m - 1} \geq \freq{m, i}$,
which implies that $m \notin \borders{p, i}$. Again, the induction assumption
implies that $q$ and $m$ will be disconnected. Consequently, $n$ will be the
first child of $m$ at some point.

Note that while moving $m$ or left siblings of $n$ to be children of $r$ we
move the current node $a$ in \updatetree to the left. Hence, there will be a
point where $a = m$ and $n$ is the first child of $m$. Theorem~\ref{thr:update}
implies that $\freq{n, i} \geq \freq{m, i}$ which is exactly the test on
Line~\ref{alg:tree:test}. Hence, $n$ will be disconnected from $m$. This proves
the lemma.
\end{proof}

\begin{lemma}
For every $c \in C$, a path in $U$ from $c$ to a child of the root node $r$ equals $\borders{c, i}$.
\end{lemma}
\begin{proof}
Fix $c \in C$ and let $\enpr{b_1}{b_M} = \borders{c, i - 1}$ and define $b_{M +
1} = i$. Theorem~\ref{thr:update} implies that there is $1 \leq N \leq M + 1$
such that $\enpr{b_1}{b_N} = \borders{c, i}$.

After adding $i$ to $T$, \updatetree will not add new nodes into the path from $c$
to $r$. Lemma~\ref{lem:prune} now
implies that the path from $c$ to $r$ will be $\enpr{b_1}{b_K}$, where $K \leq N$. If $N = 1$,
then immediately $K = 1$.  To conclude that $K = N$ in general, assume that $N
> 1$ and assume that at some point in \updatetree we have $a = b_N$ and $b =
b_{N - 1}$. Then, according to Theorem~\ref{thr:update}, the test on
Line~\ref{alg:tree:test} will fail and $b_{N - 1}$ remains as a child of
$b_N$.
\end{proof}

\begin{lemma}
Let $n$ be a node in $U$, then there is $c \in C$ such that $n \in \borders{c, i}$.
\end{lemma}
\begin{proof}
Let $m$ be a node that occurs in $T$ but not in $\btree{D[1, i], C}$.
The lemma will follow if we can show that $m$ is not in $U$.
Let $n$ be the last child of $m$. Lemma~\ref{lem:prune} 
implies that at some point $n$ will be disconnected from $m$ and we will visit $m$ when it
is a leaf, since $m \notin C$, we will delete $m$.
\end{proof}

\begin{lemma}
\label{lem:postorder}
Consider a post-order of nodes of $T = \btree{D[1, i - 1], C}$, that is,
parents and later siblings come first. Node values decrease with respect to this
order.
\end{lemma}

\begin{proof}
We will prove that the following holds:
Let $n$ be a node and let $m$ be its left sibling. Let $q$ be the smallest
child of $n$. Then $m < q$. Note that this automatically proves the lemma.

Note that $q \in C$. To prove that $m < q$, let $c \in C$ such that $m \in
\borders{c, i - 1}$.  If $c \geq q$, then since $n > m \geq c$,
Theorem~\ref{thr:share} implies that $n \in \borders{c, i - 1}$ which is a
contradiction.
Consequently, $c < q$. If $q \leq m$, then again Theorem~\ref{thr:share} implies that $m \in
\borders{q, i - 1}$ which is a contradiction. This proves that $m < q$. 
\end{proof}

\begin{lemma}
Child nodes of each node in $U$ are ordered from smallest to largest.
\end{lemma}
\begin{proof}
\updatetree modifies the tree by moving the first child of a node $a$ to be the
left sibling of $a$. This does not change the post-order of the nodes.  This
implies that, since node values decrease with respect to the post-order in $T$, they
will also decrease in $U$. This proves the lemma.
\end{proof}

%% file: paper.bbl
\begin{thebibliography}{18}
\providecommand{\natexlab}[1]{#1}
\providecommand{\url}[1]{{#1}}
\providecommand{\urlprefix}{URL }
\expandafter\ifx\csname urlstyle\endcsname\relax
  \providecommand{\doi}[1]{DOI~\discretionary{}{}{}#1}\else
  \providecommand{\doi}{DOI~\discretionary{}{}{}\begingroup
  \urlstyle{rm}\Url}\fi
\providecommand{\eprint}[2][]{\url{#2}}

\bibitem[{Basseville and Nikiforov(1993)}]{basseville:93:detection}
Basseville M, Nikiforov IV (1993) Detection of Abrupt Changes --- Theory and
  Application. Prentice-Hall

\bibitem[{Bellman(1961)}]{bellman:61:on}
Bellman R (1961) On the approximation of curves by line segments using dynamic
  programming. Communications of the ACM 4(6)

\bibitem[{Bernaola-Galv\'an et~al(1996)Bernaola-Galv\'an, Rom\'an-Rold\'an, and
  Oliver}]{bernaola-galvan:96:compositional}
Bernaola-Galv\'an P, Rom\'an-Rold\'an R, Oliver JL (1996) Compositional
  segmentation and long-range fractal correlations in dna sequences. Physical
  Review E Statistical Physics Plasmas Fluids And Related Interdisciplinary
  Topics 53(5):5181--5189

\bibitem[{Calders et~al(2007)Calders, Dexters, and
  Goethals}]{calders:07:mining}
Calders T, Dexters N, Goethals B (2007) Mining frequent itemsets in a stream.
  In: ICDM, pp 83--92

\bibitem[{Douglas and Peucker(1973)}]{douglas:73:algorithms}
Douglas D, Peucker T (1973) Algorithms for the reduction of the number of
  points required to represent a digitized line or its caricature. Canadian
  Cartographer 10(2):112--–122

\bibitem[{D\v{z}eroski et~al(2011)D\v{z}eroski, Goethals, and
  Panov}]{dzeroski:11:inductive}
D\v{z}eroski S, Goethals B, Panov P (eds)  (2011) Inductive Databases and
  Constraint-based Data Mining. Springer

\bibitem[{Gedikli et~al(2010)Gedikli, Aksoy, Unal, and
  Kehagias}]{gedikli:10:modified}
Gedikli A, Aksoy H, Unal NE, Kehagias A (2010) Modified dynamic programming
  approach for offline segmentation of long hydrometeorological time series.
  Stochastic Environmental Research and Risk Assessment 24(5)

\bibitem[{Gionis and Mannila(2003)}]{gionis:03:finding}
Gionis A, Mannila H (2003) Finding recurrent sources in sequences. In:
  Proceedings of the seventh annual international conference on Research in
  computational molecular biology, RECOMB '03, pp 123--130

\bibitem[{Gr\"{u}nwald(2007)}]{grunwald:07:minimum}
Gr\"{u}nwald P (2007) The Minimum Description Length Principle. MIT Press

\bibitem[{Haiminen and Gionis(2004)}]{haiminen:04:unimodal}
Haiminen N, Gionis A (2004) Unimodal segmentation of sequences. In: ICDM, pp
  106--113

\bibitem[{Himberg et~al(2001)Himberg, Korpiaho, Mannila, Tikanm{\"a}ki, and
  Toivonen}]{himberg:01:time}
Himberg J, Korpiaho K, Mannila H, Tikanm{\"a}ki J, Toivonen H (2001) Time
  series segmentation for context recognition in mobile devices. In: ICDM, pp
  203--210

\bibitem[{Keogh et~al(2005)Keogh, Lin, and Fu}]{keogh:05:hot-sax}
Keogh EJ, Lin J, Fu AWC (2005) {HOT SAX}: Efficiently finding the most unusual
  time series subsequence. In: ICDM, pp 226--233

\bibitem[{Kifer et~al(2004)Kifer, Ben-David, and Gehrke}]{kifer:04:detecting}
Kifer D, Ben-David S, Gehrke J (2004) Detecting change in data streams. In:
  VLDB, pp 180--191

\bibitem[{Lavrenko et~al(2000)Lavrenko, Schmill, Lawrie, Ogilvie, Jensen, and
  Allan}]{lavrenko:00:mining}
Lavrenko V, Schmill M, Lawrie D, Ogilvie P, Jensen D, Allan J (2000) Mining of
  concurrent text and time series. In: KDD Workshop on Text Mining, pp
  37--–44

\bibitem[{Palpanas et~al(2004)Palpanas, Vlachos, Keogh, Gunopulos, and
  Truppel}]{palpanas:04:online}
Palpanas T, Vlachos M, Keogh EJ, Gunopulos D, Truppel W (2004) Online amnesic
  approximation of streaming time series. In: ICDE, pp 339--349

\bibitem[{Schwarz(1978)}]{schwarz:78:estimating}
Schwarz G (1978) Estimating the dimension of a model. Annals of Statistics
  6(2):461--464

\bibitem[{Shatkay and Zdonik(1996)}]{shatkay:96:approximate}
Shatkay H, Zdonik SB (1996) Approximate queries and representations for large
  data sequences. In: ICDE, pp 536--545

\bibitem[{Terzi and Tsaparas(2006)}]{terzi:06:efficient}
Terzi E, Tsaparas P (2006) Efficient algorithms for sequence segmentation. In:
  {SIAM} Data Mining

\end{thebibliography}
